\documentclass{article}

\usepackage[preprint, nonatbib]{neurips_2019}

\usepackage[utf8]{inputenc} % allow utf-8 input
\usepackage[T1]{fontenc}    % use 8-bit T1 fonts
\usepackage{hyperref}       % hyperlinks
\usepackage{url}            % simple URL typesetting
\usepackage{booktabs}       % professional-quality tables
\usepackage{amsfonts}       % blackboard math symbols
\usepackage{nicefrac}       % compact symbols for 1/2, etc.
\usepackage{microtype}      % microtypography

\usepackage{subcaption}
\usepackage{xcolor}
\usepackage{amsmath,amssymb,amsthm,float,tkz-berge,caption,mathtools, bbold}
\usepackage{balance}
\usepackage{graphicx}
\usepackage[colorinlistoftodos]{todonotes}
\usepackage[most]{tcolorbox}
\usepackage[makeroom]{cancel}
\usepackage{algorithm}
\usepackage[noend]{algpseudocode}
\usetikzlibrary{positioning}
\tikzset{main node/.style={circle,fill=blue!20,draw,minimum size=0.5cm,inner sep=0pt},}

\tikzset{empty node/.style={circle,minimum size=0.50cm,inner sep=0pt},opacity=0.0}

\newcommand{\OO}{\mathcal{O}}

\DeclareMathOperator*{\argmax}{arg\,max}
\DeclareMathOperator*{\argmin}{arg\,min}

\newtheorem{theorem}{Theorem}
\newtheorem*{theorem*}{Theorem}
\newtheorem{lemma}{Lemma}
\newtheorem{definition}{Definition}
\newtheorem*{definition*}{Definition}
\newtheorem*{lemma*}{Lemma}

\newtheorem*{corollary*}{Corollary}
\newtheorem{claim}{Claim}
\newtheorem*{claim*}{Claim}

\newtheorem{condition}{Condition}

\allowdisplaybreaks

\title{Learning Mixtures of Graphs from Epidemic Cascades}

\author{%
  Jessica Hoffmann \\
  The University of Texas at Austin\\
  \texttt{hoffmann@cs.utexas.edu} \\
  \And
   Soumya Basu \\
  The University of Texas at Austin\\
  \texttt{basusoumya@utexas.edu} \\
  \And
  Surbhi Goel \\
  The University of Texas at Austin\\
  \texttt{surbhi@cs.utexas.edu} \\
  \And
  Constantine Caramanis \\
  The University of Texas at Austin\\
  \texttt{constantine@utexas.edu} \\
}

\begin{document}

\maketitle
\begin{abstract}
      We consider the problem of learning the weighted edges of a balanced mixture of two undirected graphs from epidemic cascades. %This is a natural setting in the context of social networks, where a post created by one user will not spread on the same graph if it is about football or if it is about politics. 
     While mixture models are popular modeling tools, algorithmic development with rigorous guarantees has lagged. Graph mixtures are apparently no exception: until now, very little is known about whether this problem is solvable. 
     
     To the best of our knowledge, we establish the first \textit{necessary and sufficient} conditions for this problem to be solvable in polynomial time on edge-separated graphs. When the conditions are met, $i.e.,$ when the graphs are connected with at least three edges, we give an efficient algorithm for learning the weights of both graphs with optimal sample complexity (up to log factors). %We extend the results to the setting in which the priors of the mixture are unknown and obtain similar guarantees. Our results naturally extend to $directed$ graphs with out-degree at least three.
     
     We give complimentary results and provide sample-optimal (up to log factors) algorithms for mixtures of $directed$ graphs of out-degree at least three, for mixture of undirected graphs of unbalanced and/or unknown priors. %, {\color{red}and we discuss the case of mixtures with $K > 2$ components.}
\end{abstract}

\section{Introduction}
Epidemic models represent spreading phenomena on an underlying graph \cite{Newman2014a}. Such phenomena include diseases spreading through a population,  security breaches in networks (malware attacks on computer/mobile networks), chains of activations in various biological networks (activation of synapses, variations in the levels of gene expression), circulation of information/influence (rumors,  news---real or fake, viral videos, advertisement campaigns) and so on. 

Most settings assume the underlying graph is known ($e.g.,$ the gene regulatory network), and focus on modeling epidemics \cite{DelVicario2016, Wu2018, Gomez-Rodriguez2013, Cheng2014, Zhao2015, liu2019ct}, detecting them \cite{Arias-castro2011, Arias-castro, Milling2015, Milling2012, Meirom2014, Leskovec2007, Khim2017}, detecting communities \cite{prokhorenkova2019learning, xie2019meta}, finding their source \cite{Shah2010, shah2012rumor, shah2010detecting, spencer2015impossibility, wang2014rumor, sridhar2019sequential, dong2019multiple}, obfuscating the source, \cite{fanti2016rumor, Fanti2014, Fanti2017}, or controlling their spread \cite{kolli2019influence, Drakopoulos2014, Drakopoulos2015, Hoffmann2018, Farajtabar2017, wang2019analysis, yan2019conquering, ou2019screen}.

The inverse problem, learning the graph from times of infection during multiple epidemics, has also been extensively studied. The first theoretical guarantees were established by Netrapalli and Sanghavi \cite{Netrapalli2012} for discrete-time models. Abrahao et al. \cite{Abrahao2013} tackled the problem for some continuous-time models, for exponential distributions. Daneshmand et al. \cite{Daneshmand2014} solved the problem for a wide class of continuous models which fit real-life diffusions. Pasdeloup et al. \cite{pasdeloup2017characterization} characterize a set of graphs for which this problem is solvable using spectral methods. Khim and Loh \cite{Khim2018} solved the problem for correlated cascades. Subsequently Trouleau et al. \cite{trouleau2019learning} showed how to learn the causal structure of Hawkes processes under synchronization noise. In parallel, Hoffmann and Caramanis \cite{hoffmann2019} showed that it is possible to robustly learn the graph from noisy epidemic cascades, even in the presence of arbitrary noise.

However, this line of research always assumes that the epidemic cascades are all of the same kind, and spread on one unique graph which entirely captures the dynamics of the spread. In reality, our observations of cascades are far more granular: different kinds of epidemics spread on the same nodes but through different mechanisms, i.e., different spreading graphs. Epidemic cascades we observe are often a {\em mixture} of different kinds of epidemics. Without knowledge of the {\em label} of the epidemic, can we recover the individual spreading graphs? For a concrete example, let us consider the ubiquitous Twitter graph. Individuals usually have multiple interests, and will share tweets differently according to the underlying topics of the tweets. For instance, two users may have aligned views on football and diametrically opposed political views, and hence may retweet each others' football tweets but not political posts. Interesting settings are those where the epidemic label (in this simple case, {\em football} and {\em politics}) is not observable. While {\em football} and {\em politics} may be easy to distinguish via basic NLP, the majority of settings will not enjoy this property (e.g., she retweets football posts relating to certain teams, outcomes or special plays). In fact, the focus on recovering the spreading graph stems precisely from the desire to study very poorly-understood epidemics where we do not understand spreading mechanisms, symptoms, etc. Examples outside the twitter realm (e.g., human epidemics with multiple spreading vectors) abound.

In such cases, applying existing techniques for estimating the spreading graph would recover the union of graphs in the mixture. For Twitter and other social networks, this is essentially already available. More problematic, this union is typically not informative enough to predict the spread of tweets, and may even be misleading. 

We address precisely this problem. We consider a mixture of epidemics that spread on two unknown weighted graphs when, for each cascade, the kind of epidemic (and hence the spreading graph) remains hidden. We aim to accurately recover the weights of both the graphs from such  cascades. 

Mixture models in general have attracted significant focus. Even for the most basic models, e.g., Gaussian mixture models, or mixed regression, rigorous recovery results have proved elusive, and only recently has there been significant progress (e.g., \cite{balakrishnan2017statistical, yi2014alternating, yi2016solving, chen2017convex, diakonikolas2018list, pmlr-v99-kwon19a, xu2016global, daskalakis2017ten, kwon2019converges}). This work reveals some similarities to prior work. For example, here too, moment-based approaches play a critical role; moreover, here too, there are conditions on separation of the two classes needed for recovery. Interestingly, however, the technical key to our work is much more combinatorial in nature, rather than appealing to more general purpose tools (like tensor decomposition, or EM). As we outline below, the crux of the proof of correctness of our algorithm is a combination of a characterization of {\em forbidden graphs} that cannot be learned, and a decomposition-reduction of a general graph to smaller subgraphs that can be learned, and later patched to produce a globally consistent solution. 

\subsection{Contributions}
To the best of our knowledge, this is the first paper to study the inverse problem of learning mixtures of weighted undirected graphs from epidemic cascades. We address the following questions:

\noindent\textbf{Recovery:} Under the assumption that the underlying graphs are connected, have at least three edges, and under some separability condition (detailed in the next section), we prove the problem is solvable and give an efficient algorithm to recover the weights of $any$ mixture of connected graphs with equal priors on the same set of vertices.

\noindent\textbf{Identifiability:} We show the problem is not solvable in polynomial time of one of the condition mentionned above is violated. The problem is unidentifiable  when one of the graphs of the mixture has a connected component with less than three edges. Moreover, there exist (many) graphs which violate the separability condition, and for which any algorithm would require at least exponential (in the number of nodes) sample complexity.

\noindent\textbf{Sample Complexity:} We prove a lower bound on the sample complexity of the problem, and show that our algorithm always matches the lower bound up to log factors in terms of the number of nodes $N$. It also matches the bound exactly in terms of the dependency in the separation parameter $\frac{1}{\Delta}$ if the graphs have min-degree at least 3.

\noindent\textbf{Extensions:} We give similar guarantees for the case of directed graphs of min-degree at least 3, of undirected graphs with unbalanced and/or unknown mixtures priors. %, {\color{red} and we give sufficient conditions for numerically solving the setting with more than 2 mixtures.}

\section{Preliminaries}
%\subsection{Model for Sample Generation}
We consider an instance of the \textit{independent cascade model} \cite{Goldberg2001, Kempe2003}. We observe independent epidemics spreading on a mixture of two graphs. In this section, we specify the dynamics of the spreading process, the observation model, and the learning task.

\subsection{Mixture Model}
We consider two $weighted$ graphs $G_1 = (V, E_1)$ and $G_2 = (V, E_2)$ on the \textit{same set of vertices} $V$. Unless specified otherwise, the graphs considered are $undirected$: $p_{ij}=p_{ji}$ and $q_{ij} = q_{ji}$. Note that $p_{ij}$ ($q_{ij}$) is 0 if there is no edge between $i$ and $j$ in $G_1$ ($G_2$).

We say that the mixture is \underline{$\Delta$-\textit{separated}} if: $$\min\limits_{(i,j) E_1 \cap E_2} |p_{ij} - q_{ij}| \geq \Delta > 0.$$
We denote the minimum edge weight by ${p_{min} := \min\limits_{(i,j) \in E_1}  \min\limits_{(k,l) \in E_2} \min (p_{ij},  q_{kl})   > 0}.$ 

%We assume the following from hereon, unless stated otherwise:

% 1.  Both the graphs are \underline{\em undirected}, $i.e.$ $p_{ij} = p_{ji}$, and $q_{ij} = q_{ji}$.
 
% 2. The \underline{\em minimum weight} of an edge is  positive, $i.e.$, 
% ${p_{min} := \min\left(\min\limits_{(i,j) \in E_1} p_{ij},  \min\limits_{(i,j) \in E_2} q_{ij}\right)  > 0}.$
 
% 3.  For an edge, $(i,j) \in E_1 \cap E_2$,  its weights are \underline{\em well-separated},  $\Delta := \min\limits_{(i,j) E_1 \cap E_2} |p_{ij} - q_{ij}| > 0$.

\subsection{Dynamics of the Spreading Process}
We observe $M$ independent identically distributed epidemic cascades, which come from the following generative model. 

\paragraph{Component Selection:}  At the start of a cascade, an i.i.d.\ Bernoulli random variable $b \in \{1,2\}$ with parameter $\alpha$ ($\Pr[b = 1] = \alpha$) decides the component of the mixture, i.e., the epidemic spreads on graph $G_b$. We say that the mixture is \underline{balanced} if $\alpha = 0.5$, and we call $\alpha$ and $1-\alpha$ the \underline{priors} of the mixture. Unless specified otherwise, the results presented are for balanced mixtures.
%For ease of exposition, we first present our results for the balanced case and then the generalization to unbalanced mixtures with unknown $\alpha$.
%We first consider the setting where $\{L^m\}$ are balanced Bernoulli random variable. We then consider the extension of this model to a mixture of two graphs with unknown prior $\alpha$: for each epidemic $m$, the label is given by independent random variable $L_{\alpha}^{m}$, such that $\Pr(L_{\alpha}^{m} = 1) = \alpha = 1 - \Pr(L_{\alpha}^{m} = 2)$. 

%\textbf{Component Selection:}  For the $m$-th cascade the i.i.d. Bernoulli random variable $L^{m}$ decides the component of the mixture, i.e. its source graph is $G^m = G_k$ if $L^{m} = k$ for $k =1, 2$.  We first consider the setting where $\{L^m\}$ are balanced Bernoulli random variable. We then consider the extension of this model to a mixture of two graphs with unknown prior $\alpha$: for each epidemic $m$, the label is given by independent random variable $L_{\alpha}^{m}$, such that $\Pr(L_{\alpha}^{m} = 1) = \alpha = 1 - \Pr(L_{\alpha}^{m} = 2)$. 

 \paragraph{Epidemic Spreading:} Once the component of the mixture $G_b$ is fixed, the epidemic spreads in \textit{discrete time} on graph $G_b$ according to a regular one-step {\em Susceptible} $\rightarrow$ {\em Infected} $\rightarrow$ {\em Removed} (SIR) process \cite{Netrapalli2012, hoffmann2019}. At $t=0$, the epidemic starts on a $unique$ source, chosen uniformly at random among the nodes of $V$. The source is in the Infected state, while all the other nodes are in the Susceptible state. Let $I_t$ (resp $R_t$) be the set of nodes in the Infected (resp.\ Removed) state at time $t$. At each time step $t \in \mathbb{N}$, all nodes in the Infected state try to infect their neighbors in the Susceptible state, before transitioning to the Removed state during this same time step (\textit{i.e.}, $R_{t+1} = R_{t} \cup I_t$) \footnote{Once a node is in the Removed state, the spread of the epidemic proceeds as if this node were no longer on the graph.}. If $i$ is in the Infected state at time $t$, and $j$ is in the Susceptible state at the same time ($i.e$ $i\in I_t, j \in S_t$), then $i$ infects $j$ with probability $p_{ij}$ if $b = 1$, and $q_{ij}$ if $b = 2$. Note that multiple nodes in the Infected state can infect the same node in the Susceptible state. The process ends at the first time step such that all nodes are in the Susceptible or Removed state ({\em i.e.}, no node is in the Infected state).  %The number of Removed nodes at the end of a cascade is called the \textbf{size of this cascade}.

% \textbf{Epidemic Spreading:} Once the component of the mixture $G^m$ is fixed, the epidemic spreads in \textit{discrete time} on graph $G^m$ according to a regular one-step {\em Susceptible} $\rightarrow$ {\em Infected} $\rightarrow$ {\em Removed} (SIR) process \cite{Netrapalli2012, hoffmann2019}. At $t=0$, epidemic $m$ starts on a $unique$ source, chosen uniformly at random among the nodes of $V$. The source is in the Infected state, while all the other nodes are in the Susceptible state. Let $I_t$ (resp $R_t$) be the set of nodes in the Infected (resp. Removed) state at time $t$. At each time step $t \in \mathbb{N}$, all nodes in the Infected state try to infect their neighbors in the Susceptible state, before transitioning to the Removed state during this same time step (\textit{i.e.}, $R_{t+1} = R_{t} \cup I_t$) \footnote{Once a node is in the Removed state, the spread of the epidemic proceeds as if this node were no longer on the graph.}. If $i$ is in the Infected state at time $t$, and $j$ is in the Susceptible state at the same time ($i.e$ $i\in I_t, j \in S_t$), then $i$ infects $j$ with probability $p_{ij}$ if $G^m = G_1$, and $q_{ij}$ if $G^m = G_2$. Note that multiple nodes in the Infected state can infect the same node in the Susceptible state.

One realization of such a process from randomly picking the component of the mixture and the source at $t=0$ to the end of the  process is called a \underline{cascade}.

\subsection{Observation Model} 
For each cascade we do not have the knowledge of the underlying component, that is, we do not observe $b$ and we treat this as a missing label. For each cascade, we have access to the complete list of infections: we know which node infected which node at which time (one node can have been infected by multiple nodes). This list constitutes a \underline{\em sample} from the underlying mixture model. 

%\textbf{Observation:} For each cascade $m$ we do not have the knowledge of the underlying component $$,  and we treat this as a missing label. For each cascade, we have access to the complete list of infections: we know which node infected which node at which time (one node can have been infected by multiple nodes). Such a list is called a \underline{\em sample}. 

\subsection{Learning Objective}
Our goal is to learn the weights of all the edges of the underlying graphs of the mixture, up to precision $\epsilon < \min(\Delta, p_{min})$. Specifically, we want to provide $\hat{p}_{ij}$ and $\hat{q}_{ij}$ for all vertex pairs $i, j \in V$ such that $\max_{i,j \in V^2} \max( |p_{ij} - \hat{p}_{ij}|, |q_{ij} - \hat{q}_{ij}|) < \epsilon$.

\subsection{When is this problem solvable?}\label{sec:when}
\begin{figure}
	\centering
			\begin{subfigure}[t]{0.15\textwidth}
		\centering
		\begin{tikzpicture}
		\node[main node, fill=white] (0) {j};
		\node[main node, fill=white] (1) [right = 0.25cm of 0] {i};
		
		\path[-]
		(1) edge node {} (0);
		\end{tikzpicture}
		\caption{One edge}\label{fig:one}
	\end{subfigure}
		\begin{subfigure}[t]{0.15\textwidth}
		\centering
		\begin{tikzpicture}
		\node[main node, fill=white] (0) {j};
		\node[main node, fill=white] (1) [right = 0.25cm of 0] {i};
		\node[main node, fill=white] (2) [above right = 0.25cm and 0.1cm of 1]  {k};
		
		\path[-]
		(1) edge node {} (0)
		(1) edge node {} (2);
		\end{tikzpicture}
		\caption{Two edges}\label{fig:two}
	\end{subfigure}
		\begin{subfigure}[t]{0.15\textwidth}
	\centering
	\begin{tikzpicture}
	\draw (0,0) ellipse (0.5cm and 1cm);
	\node[empty node] (0) {A};
	\draw (1.5,0) ellipse (0.5cm and 1cm);
	\node[empty node] (1) [right = 1cm of 0] {B};
	\end{tikzpicture}
	\caption{Disconnected components}\label{fig:disconnected}
\end{subfigure}
\caption{Unsolvable structures}
\end{figure}
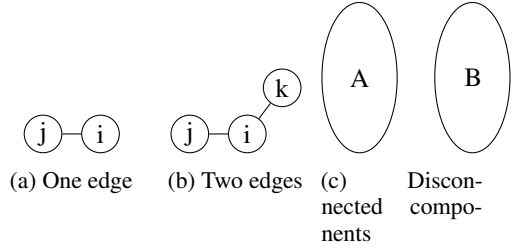
Prior to presenting our main results, we offer some intuition. We show that it is not always possible to learn the weights of both components of the mixture, even for settings that appear deceitfully easy. 

Indeed, it is impossible to learn the graph on two nodes $i$ and $j$, with only one directed edge from $i$ to $j$ (see Figure \ref{fig:one}). To see this, consider a balanced mixture, for which edge $(i,j)$ has weight $\beta$ in $G_1$, and  weight $1-\beta$ in $G_2$, then $i$ will infect $j$ half of the time, independently of the value of $\beta$. This shows that we cannot recover the original weights, and the mixture problem is not solvable. If we add another edge, and $i$ is now connected to a new node $k$ (see Figure \ref{fig:two}), the problem is still not solvable (see Supplementary Material).

Surprisingly, if $i$ has a third neighbor $l$ (see Figure \ref{fig:smallstar}), it becomes possible to learn the weights of the mixture. Learning this local structure is one of the main building blocks of our algorithm.

One could think that four nodes are needed for this problem to be solvable. However, we can learn the edges of a triangle (see Figure \ref{fig:triangle}). Similarly, the intuition that nodes need to be of degree at least three is misleading. If a line has more than three nodes (see Figure \ref{fig:line}), it is solvable. The line on four nodes is the other local structure which forms the foundation of our algorithm.

On the other end, the setting for which there exists (at least) two parts of the graph $A$ and $B$ for which cascades never overlap is a general unsolvable setting (see Figure \ref{fig:disconnected}). We write $A_i = A \cap E_i, B_i = B \cap E_i$. Let $E'_1 = A_1 \cup B_2$, $E'_2 = A_2 \cup B_1$. We notice a mixture spreading on edges $E_1$ and $E_2$ yields the same cascade distribution as a mixture on $E'_1$ and $E'_2$. Therefore, the solution is not unique. 

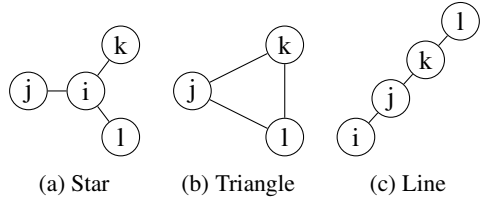
\begin{figure}
	\centering
	\begin{subfigure}[t]{0.15\textwidth}
		\centering
		\begin{tikzpicture}
		\node[main node, fill=white] (0) {j};
		\node[main node, fill=white] (1) [right = 0.25cm of 0] {i};
		\node[main node, fill=white] (2) [above right = 0.25cm and 0.1cm of 1]  {k};
		\node[main node, fill=white] (3) [below right = 0.25cm and 0.1cm of 1] {l};
		
		\path[-]
		(1) edge node {} (0)
		(1) edge node {} (2)
		(1) edge node {} (3);
		\end{tikzpicture}
		\caption{Star}\label{fig:smallstar}
	\end{subfigure}
	\begin{subfigure}[t]{0.15\textwidth}
		\centering
		\begin{tikzpicture}
		\node[main node, fill=white] (0) {j};
		\node[empty node] (1) [right = 0.25cm of 0] {};
		\node[main node, fill=white] (2) [above right = 0.25cm and 0.1cm of 1]  {k};
		\node[main node, fill=white] (3) [below right = 0.25cm and 0.1cm of 1] {l};
		
		\path[-]
		(0) edge node {} (2)
		(2) edge node {} (3)
		(3) edge node {} (0);
		\end{tikzpicture}
		\caption{Triangle}\label{fig:triangle}
	\end{subfigure}
	\begin{subfigure}[t]{0.16\textwidth}
	\centering
	\begin{tikzpicture}
	\node[main node, fill=white] (0) {l};
	\node[main node, fill=white] (1) [below left= 0.15cm and 0.1cm of 0] {k};
	\node[main node, fill=white] (2) [below left= 0.15cm and 0.1cm of 1]  {j};
	\node[main node, fill=white] (3) [below left= 0.15cm and 0.1cm of 2] {i};
	
	\path[-]
	(3) edge node {} (2)
	(2) edge node {} (1)
	(1) edge node {} (0);
	\end{tikzpicture}
	\caption{Line}\label{fig:line}
\end{subfigure}
	\caption{Solvable local structure}\label{fig:solvable}
\end{figure}

The three simple shapes in Figure \ref{fig:solvable} form the core of this paper. Our key insight is in showing that any graph that can be built up using these three building blocks (i.e., each node belongs in at least one of these structures) is solvable. This effective decomposition succeeds in reducing a general problem to  a small number of sub-problems, for which we provide a solution. 

%In this work we show that the class of undirected graphs for which this problem is solvable is exactly the class of graphs which can be generated from the three shapes in Figure \ref{fig:solvable}, i.e., each node belongs in at least one of these structures. For graphs on more than three nodes, we can even restrict ourself to two of these structures. It is remarkable that whether this complex mixture problem can be solved can be reduced to verifying this simple local property. 

%\section{Impossibility results}
%\input{files/unindentify}

\section{Main Results}
In this section we present our main results on the impossibility and recoverability of edge weights for a balanced mixture.
%The proofs of Theorem \ref{thm:mainnotexist} and Theorem \ref{thm:lowerbound} are deferred to the supplementary material, whereas we present proof sketch for Theorem \ref{thm:mainexist} for the balanced mixture case along with an algorithm.

\subsection{Balanced Mixture of Undirected Graphs}
\paragraph{Impossibility Result Under Infinite Samples}
\begin{condition}\label{ass:necessary}
	The graph $G = (V, E_1 \cup E_2)$ is connected and has at least three edges: $|E_1 \cup E_2| \ge 3$.
\end{condition}
%The above condition is a necessary condition for recovery. 

\begin{claim}\label{cl:mainnotexist}
	Suppose Condition \ref{ass:necessary} is violated. Then it is impossible to recover the edge weights corresponding to each graph (even with infinite samples).
\end{claim}

\paragraph{Impossibility Result Under Polynomial Samples}
\begin{condition}\label{ass:main}
	The mixtures  in the graph $G = (V, E_1 \cup E_2)$ are well-separated, that is, $\Delta > 0$.
\end{condition}

\begin{claim}\label{cl:mainnotexist2}
	Suppose Condition \ref{ass:main} is violated. Then there exists (many) graphs for which we need at least exponential (in the number of nodes $N$) samples to recover the edge weights.
\end{claim}

\paragraph{Recoverability Result with Finite Samples}
\begin{theorem}\label{thm:mainexist}
Suppose Conditions \ref{ass:necessary} and \ref{ass:main} are true. Then there exists an algorithm that runs on epidemic cascades over a balanced mixture of two undirected, weighted graphs $G_1 = (V, E_1)$ and $G_2 = (V, E_2)$, and recovers the edge weights corresponding to each graph up to precision $\epsilon$ with probability at least $1-\delta$, in time $O(N^2)$ and sample complexity $O\left(\frac{N}{\epsilon^2\cdot \Delta^4} \log (\tfrac{N}{\delta} )\right)$, where $N = |V|$.
\end{theorem}

\paragraph{Remark on Partial Recovery:} An important element of our results is that if Conditions~\ref{ass:necessary} and \ref{ass:main} are not satisfied for the entire graph we can still recover the biggest subgraph which follows these conditions. In particular, if the graph we obtain by removing all non $\Delta$-separated edges is still connected, we can detect and learn all the edges of the graph (see Supplementary Material for more details). This is important, as it effectively means that we are able to learn the mixtures in the parts of the graph that matter most. On a practical note, this also means that our algorithm is resistant to the presence of bots in the network that retweet everything indifferently.

\subsection{Extensions}
\paragraph{Extension to Directed Graphs}
Interestingly, the techniques used to prove the theorem above can be immediately applied to learn mixtures of directed graphs of out-degree at least three (see Supplementary Material for complete proof). Note that the better dependency in $\frac{1}{\Delta}$ comes from the assumption on the degree \footnote{This immediately implies a better dependency in $\frac{1}{\Delta}$ for learning undirected graphs of minimum degree three.}. Since many applications on social networks can ignore nodes of out-degree less than three, as thoses nodes have very little impact on any diffusion phenomena, this result is of independent interest:

\begin{theorem}\label{thm:maindirected}
	Suppose Conditions \ref{ass:necessary} and \ref{ass:main} are true. Then there exists an algorithm that runs on epidemic cascades over a balanced mixture of two \textbf{directed}, weighted graphs of minimum out-degree three $G_1 = (V, E_1)$ and $G_2 = (V, E_2)$, and recovers the edge weights of each graph up to precision $\epsilon$  with probability at least $1-\delta$, in time $O(N^2)$ and sample complexity $O\left(\frac{N}{\epsilon^2\cdot \Delta^2} \log (\tfrac{N}{\delta} )\right)$, where $N = |V|$.
\end{theorem}

\paragraph{Extension to Unbalanced/Unknown Priors }
If the mixture is unbalanced, but the priors are known, we can adapt our algorithm to learn the mixture \textit{under the same conditions as above}, at the price of a higher dependency in $\frac{1}{\Delta}$. If the priors are unknown, we can only recover graphs of min-degree at least three.
% \paragraph{Remarks on Condition~\ref{ass:necessary} and \ref{ass:main}:} %We now discuss the significance of the assumptions in the above theorem,  in particular we explain if they come from impossibility results or restrictions in our algorithm.

% $\bullet$ Condition~\ref{ass:necessary} is necessary for identifiability.
% Indeed, the case of one or two edges has been discussed in Section \ref{sec:when}, as well as the case of disconected components. For sufficiency, the case of the triangle (three nodes and three edges) is treated in Supplementary Material \ref{app:triangle}, while the rest of the paper provides a proof for graphs on four nodes or more.

%\paragraph{Extension to Mixture of $K \geq 2$ Graphs}
%{\color{red} Although we do not provide guarantees for a mixture of $K\geq2$ graphs, we note that for \textit{directed graphs of of minimum out-degree $2K -1$}, we can reformulate the problem as a Semi-Definite Programming problem, for which we know the feasible set is not empty. We can therefore numerically solve this problem.} 

\subsection{Lower Bounds}
We provide two lower bounds, one for undirected graphs, one for directed graphs, for mixtures of two graphs.

\begin{theorem}\label{thm:lowerbound}
	When learning the edge weights of a balanced mixture on two $\Delta$-separated graphs on $N$ nodes up to precision $\epsilon < \Delta$, we need:
	\begin{enumerate}
		\item $\Omega \left( \frac{N}{\Delta^2} \right)$ samples for undirected graphs, which proves our algorithm is optimal in $N$ up to log factors in this setting.
		\item $\Omega \left( N\log(N) + \frac{N \log\log(N)}{\Delta^2} \right)$ samples for directed graphs of minimum out-degree three, which proves our algorithm has optimal dependency in $N$ and in $\frac{1}{\Delta^2}$ in this setting.
	\end{enumerate}
\end{theorem}

\section{Balanced  Mixture of Undirected Graphs}
In this section, we provide our main algorithm (Algorithm \ref{algo:generalGraph}) that recovers the edge weights on the graph under the conditions presented in Theorem \ref{thm:mainexist}.
% We first present our main algorithm followed by an overview. Subsequently we discuss the sub-procedures in detail and conclude with a proof of correctness. 

% We focus on the setting of balanced mixture priors for ease of exposition.
% We refer the reader to Supplementary \ref{app:unknownAlpha} for the results on the unbalanced mixture case. 

% \subsection{Main Algorithm}
% We can now turn to the main algorithm (see Algorithm \ref{algo:generalGraph}). This algorithm first uses the algorithms \textsc{LearnEdges} to learn all the edges of $E_1 \cup E_2$.  It then calls \textsc{Learn2Nodes} as an initialization, which learns all weights of the edges connected to two initial nodes (see Supplemental Materials, Algorithm \ref{app:learn2Nodes}). The set $S$, which only contains nodes for which the weights of each incident edge has been learned, is initialized with these two nodes. 

% Then, at each iteration, of our algorithm we pick one node connected to $S$. If this node is a star vertex, we learn all the weights of its neighborhood using  \textsc{LearnStar}. If this node is a line vertex, we we learn all the weights of its neighborhood using \textsc{LearnLine}. If this node has degree one, we have already learned all the weights of the edges connected to it, so we proceed without doing anything. Finally, since we have learned all the weights of the edges connected to this new node, we can add it to $S$. We keep growing $S$ until $S = V$.

\begin{algorithm}
    \caption{Learn the weights of undirected edges}\label{algo:generalGraph}
    \hspace*{\algorithmicindent} \textbf{Input} Vertex set $V$\\
    \hspace*{\algorithmicindent} \textbf{Output} Edge weights for the two epidemics graphs
    \begin{algorithmic}[1]
        \State $E \leftarrow \textsc{LearnEdges}(V)$ \Comment{\textbf{Learn the edges}}
        \State $S, W \leftarrow \textsc{Learn2Nodes}(V, E)$ \Comment{\textbf{Initialize}}
        \While {$S \neq V$}
        \State Select $u \in S, v \in V \backslash S$ such that $(u,v) \in E$
        \If {deg(v) $\geq$ 3} \Comment{\textbf{Use star primitive}}
        \State $W \leftarrow W \cup \textsc{LearnStar}(v, E, W)$
        
        \EndIf
        \If{deg(v) = 2} \Comment{\textbf{Use line primitive}}
        \State Set $w \in S$ such that $(u, w) \in E$
        \State Set $t \in V$ such that $(v, t) \in E$ and $t \ne u$
        \If {$t \not \in S$} 
        \State $W \leftarrow W \cup \textsc{LearnLine}(t, v, u, w, S, W)$
        \EndIf
        \EndIf
        \State $S \leftarrow S \cup \{v\}$
        \EndWhile
        % \EndIf
        \State {\bf Return} $W$
    \end{algorithmic}
\end{algorithm}

\subsection{Overview of Algorithm \ref{algo:generalGraph}}
First, the algorithm learns the edges of the underlying graph using the procedure \textsc{LearnEdges}. To detect whether an edge $(u,v)$ exists in $E_1 \cup E_2$, we use a simple estimator (Section \ref{sec:struct}). This also provides us with the degree of each node with respect to $E_1 \cup E_2$.

With the knowledge of the structure of the graph, to learn the edge weights adjacent to a node, our algorithm uses two main procedures \textsc{LearnStar} and \textsc{LearnLine}. If a node is of degree at least three (e.g., node $u$ in Figure~\ref{fig:star}), procedure \textsc{LearnStar} recovers all the edge weights (i.e., the weights of the two mixtures for these edges) adjacent to this node {\em independently} of the rest of the graph. Otherwise, if a node is of degree two (e.g., node $u$ in Figure~\ref{fig:line}), procedure \textsc{LearnLine} learn all the edge weights adjacent to this node {\em independently}. Both procedures use carefully designed estimators that exploit the respective structures. We present the above estimators for balanced mixtures in  (Section \ref{sec:estimators}). We require Condition~\ref{ass:main} for the existence of the proposed estimators.

Our main algorithm maintains a set of {\em learned nodes}. A node is a {\em learned node} if the weights for all the edges adjacent to it have been learned. The algorithm begins with learning two connected nodes (two nodes having an edge in between) using procedure \textsc{Learn2Nodes}. Next it proceeds iteratively, by learning the weights of the edges connected to one {\em unlearned neighbor} of the {\em learned nodes} using the two procedures discussed above. The algorithm terminates, when all the nodes in $V$ are learned.  In Theorem~\ref{thm:mainnotexist}, we show that under Conditions \ref{ass:necessary}, it is possible to iteratively learn all the nodes in $V$. 
%\textbf{Unbalanced mixture:} We require different estimators for the known but general $\alpha$ case, which are deferred to Supplementary \ref{app:unknownAlpha}. We note that Algorithm \ref{algo:generalGraph} works for general $\alpha$ as well with the appropriate estimators.

\subsection{Learning Edges, Star and Line Vertices}\label{sec:estimators}
%In this section, we show how we recover the weights for local structures using moment matching methods. The key idea is to express probabilities of some observable events in the cascades as a function of the weights. These provide us with a system of polynomial equations. Such systems are very hard to solve in general. Our main contribution is to introduce a set of estimators for which this system can be decoupled into multiple systems of low degree polynomial, for which we find a closed-form solution. 

In this section, we show how we recover the weights for local structures using moment matching methods. Our proof relies on a few crucial ideas. First, we introduce local estimators, which can be computed from observable events in the cascade, and are polynomials of the weights of the mixture. General systems of polynomial equations are hard to solve. However, we found ways of combining these specific estimators to decouple the problem, and obtain $\OO(|E|)$ systems of six polynomial equations of maximum degree three, with six unknowns. Finally, we show how to elegantly get a closed-form solution for these systems.

%Under Condition \ref{ass:necessary} and \ref{ass:main}, we are able to recover the weights. We present the main results and intuition while deferring the proofs to supplementary material.  

\subsubsection{Learning the Edges in $E_1 \cup E_2$}\label{sec:struct}
We recall that $I_0$ is the random variable indicating the set containing the unique source of the epidemic for a cascade.  If an epidemic cascade starts from node $u$, then for any node $a$ that is infected in time step $1$ there is an edge $(u,a) \in E_1\cup E_2$. This provides us with the average weight of the edge $(u,a)$ as $X_{ua}$, 

\begin{claim} \label{cl:trivialXab}
	If $u$ and $a$ are two distinct nodes of $V$ such that $(u,a) \in E_1 \cup E_2$, then: 
	$$X_{ua} :=  \Pr(u \rightarrow a ~|~ u \in I_0) =  \frac{p_{ua} + q_{ua}}{2} \geq \frac{p_{min}}{2}.$$
	Furthermore, there exists an edge between $u$ and $a$ in $E_1 \cup E_2$, if and only if $X_{ua} \geq \frac{p_{min}}{2} > 0$.
\end{claim}

The above claim can be leveraged to design algorithm \textsc{LearnEdges}, which takes as inputs all the $X_{ua}$ for all pairs $(u,a)$, and returns all the edges of $E_1 \cup E_2$ (See Supplemental Materials).

\textbf{Conditioning on Source Node:} We notice that the expression of $X_{ua}$ is a function of weights of edges $(u, a)$. Here, the conditioning on the event that $u\in I_0$ plays an important role. Indeed, if the source had been any other node than $u$, then the probability that $a$ was not removed when $u$ is infected would have depended on the (unknown) weights of the paths connecting the source and node $a$. We could not have obtained a simplified form as above.

\subsubsection{Star vertex}

\begin{figure}
\centering
\begin{tikzpicture}
    \node[main node, fill=white] (0) {$a$};
    \node[main node, fill=pink] (1) [right = 0.5cm of 0] {$u$};
    \node[main node, fill=white] (2) [above right = 0.5cm and 0.2cm of 1]  {$b$};
    \node[main node, fill=white] (3) [below right = 0.5cm and 0.2cm of 1] {$c$};

    \path[draw,thick]
    (0) edge node {} (1)
    (1) edge node {} (2)
    (1) edge node {} (3);
\end{tikzpicture}
\caption{A star vertex $u$, with edges $(u,a), (u,b)$ \\and $(u,c)$ in $E_1 \cup E_2$.}\label{fig:star}
\end{figure}
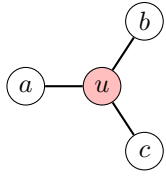

A \underline{star vertex} is a vertex $u \in V$ of degree at least three in $E_1 \cup E_2$ (Figure~\ref{fig:star}). We consider\\
$Y_{ua, ub}$: the probability the star vertex $u$ infects neighbors $a$ and $b$, conditioned on $u$ being the source vertex. 

\begin{claim}\label{claim:eststar}
For $u$ and $a$, $b$ and $c$ as in Figure \ref{fig:star}:
 $$Y_{ua, ub} = \Pr(u \rightarrow a, u \rightarrow b ~|~ u \in I_0) = \frac{p_{ua}p_{ub} + q_{ua}q_{ub}}{2}.$$
\end{claim}

Again, we emphasize that the conditioning places a crucial role in the product form of $Y_{ua, ub}$, similar to $X_{ua}$. Further, we make a key observation that for $i \neq j \in \{a,b,c\}$,
\begin{align} \label{eq:star}
Y_{ui, uj} - X_{ui}X_{uj} &= \frac{(p_{ui} - q_{ui})(p_{uj} - q_{uj})}{4}.
\end{align}
This directly leads to the closed form expressions for the weights of the edges adjacent to the star vertex $u$. 
\begin{lemma}\label{lem:star}
Suppose Conditions \ref{ass:necessary} and  \ref{ass:main} are true and $\alpha =1/2$. Let $s_{ua} \in \{-1,1\}$. The weight of any edge $(u,a)$ connected to a star vertex $u$, with distinct neighbors $a$, $b$ and $c$ in $E_1 \cup E_2$, is given by:
\begin{align*}
p_{ua} &= X_{ua} + s_{ua} \sqrt{\frac{(Y_{ua,ub} - X_{ua}X_{ub})(Y_{ua,uc} - X_{ua}X_{uc})}{Y_{ub,uc} - X_{ub}X_{uc}}},\\ 
q_{ua} &= X_{ua} - s_{ua} \sqrt{\frac{(Y_{ua,ub} - X_{ua}X_{ub})(Y_{ua,uc} - X_{ua}X_{uc})}{Y_{ub,uc} - X_{ub}X_{uc}}}.
\end{align*} 
Furthermore, any two sign $s_{ui}$ ans $s_{uj}$, for $i\neq j$ and $i,j\in \{a,b,c\}$, satisfy
$s_{ui}s_{uj} = \mathsf{sgn}(Y_{ui,uj} - X_{ui}X_{uj}).$ 
\end{lemma}
\textbf{Resolving Sign Ambiguity}: We note that even though the $p_{ui}$ and $q_{ui}$ has associated signs $s_{ui}$ for all $i\in \{a,b,c\}$, there is no ambiguity as fixing one sign, say $s_{ua}$ fixes the sign for all. Using the fact that  $s_{ui}= \mathsf{sgn}(p_{ui} - q_{ui})$ for all $i \in \{a,b,c\}$ we obtain the lemma.
 
\paragraph{\textsc{LearnStar}:}  The algorithm \textsc{LearnStar}, takes as input a star vertex $u$, the set of edges of $E_1 \cup E_2$ (which can be recovered using \textsc{LearnEdges}), and all the $X_{ui}$ and $Y_{ui,uj}$ for all $(i,j)$ distinct neighbors of $u$, and returns all the weights of the edges connected to $u$ in both mixtures using the above closed form expressions (See Supplemental Materials).% \ref{app:learnStar}).

\subsubsection{Line vertex}
\begin{figure}
	\centering
	\begin{tikzpicture}
	\node[main node, fill=white] (0) {$a$};
	\node[main node, fill=pink] (1) [right = 0.5cm of 0] {$u$};
	\node[main node, fill=white] (2) [right = 0.5cm of 1]  {$b$};
	\node[main node, fill=white] (3) [right = 0.5cm of 2] {$c$};
	\path[draw,thick]
	(0) edge node {} (1)
	(1) edge node {} (2)
	(2) edge node {} (3);
	\end{tikzpicture}
	% \includegraphics[width=0.4\linewidth]{line.jpg}
	% \caption{line}\label{fig:line}
	\caption{A line vertex $u$, with edges $(u,a), (u,b)$ and $(b,c)$ in $E_1 \cup E_2$.}\label{fig:line}
\end{figure}
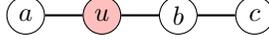
We now consider a node $u$ that has degree exactly two in $E_1 \cup E_2$ and forms a line structure.  Specifically, let $u, a, b$ and $c$ be four distinct nodes of V, such that $(a,u), (u,b)$ and $(b,c)$ belong in $E_1 \cup E_2$. We call such a node $u$ a \underline{line vertex} (see Figure \ref{fig:line}).

To recover the weights of the edges adjacent to a line vertex, only considering events in the first two timesteps is insufficient. Unlike a star vertex, we have only one second moment for a line vertex. In Figure \ref{fig:line}, we can not have $u$ as a source and, for $i \in \{a,b\}$, nodes $i$ and $c$  infected in the first timestep, i.e., $Pr(u\to i, u \to c| u \in I_0) = 0$. 

We circumvent the problem by considering:\\
1) $Y_{ub, bc}$: the probability of the event when (in Figure \ref{fig:line}) $u$ infects only $b$, and in turn $b$ infects $c$, conditioned on $u$ being the source. \\
2) $Z_{ua, ub, bc}$: the probability of the event when (in Figure \ref{fig:line}) $u$  infects both $a$ and $b$, and in turn $b$ infects $c$, conditioned on $u$ being the source. 

\begin{claim}\label{claim:estline}
For a line vertex $u$ and nodes $a$, $b$ and $c$ as in Figure \ref{fig:line}:
\begin{align*}
    &\bullet Y^|_{ua,ub} = Pr(u \rightarrow a, u \rightarrow b ~|~ u \in I_0) 
    = \tfrac{p_{ua}p_{ub} + q_{ua}q_{ub}}{2},\\
    &\bullet Y^|_{ub,bc} = \Pr(u \rightarrow b, b \rightarrow c ~|~ u \in I_0) 
    = \tfrac{p_{ub}p_{bc} + q_{ub}q_{bc}}{2},\\
    &\bullet Z^|_{ua,ub,bc} = \Pr(u \rightarrow a, u \rightarrow b, b \rightarrow c~|~ u \in I_0) \\
    &= \tfrac{p_{ua}p_{ub}p_{bc} + q_{ua}q_{ub}q_{bc}}{2}.
\end{align*}
\end{claim}
The result for $Y^|_{ua,ub}$ is similar to Claim \ref{claim:eststar} and uses the conditioning of $u$ being source crucially. However, the proof for $Y^|_{ub,bc}$ and $Z^|_{ub,bc,bc}$ is different. Additional to $u$ being the source, the proof also rely on the fact that {$u$ is of degree 2},  implying $p_{uc} = q_{uc} = 0$.

We note that there is no $Y^|_{bc, ua}$ present as $c$ can not be infected if $b$ is not. So we cannot directly replicate the star vertex. Let us define $R^| := X_{ua}X_{bc} + \frac{Z^|_{ua, ub, bc} - X_{ua}Y^|_{ub, bc} - X_{bc}Y^|_{ua,ub}}{X_{ub}}$. We notice the following interesting equality that acts as a surrogate for $(Y^|_{ua, bc} - X_{ua}X_{bc})$. 
\begin{align}\label{eq:line}
R^| &= \frac{1}{4}(p_{ua} - q_{ua})(p_{bc} - q_{bc}).
\end{align}
Similar to Lemma~\ref{lem:star} we now obtain the closed form expressions for the weights associated with the 
line vertex $u$. For unifying notations we define $Y^|_{bc, ua} := (R^| + X_{bc} X_{ua})$ (it has no probabilistic interpretation).

\begin{lemma}\label{lem:line}
Suppose Conditions \ref{ass:necessary} and  \ref{ass:main} holds and $\alpha = 1/2$,  then for $s_{ua}, s_{ub}$ and $s_{bc}$ in $\{-1,1\}$ the weights of the edges for a line structure are then given by: 
\begin{align*}
	&\forall (e1, e2, e3) \in \{(ua, ub, bc) , (ub, bc, ua), (bc, ua, ub)\},\\
    p_{e1} &= X_{e1} + s_{e1} \sqrt{\frac{(Y^|_{e1,e2} - X_{e1}X_{e2})(Y^|_{e3, e1} - X_{e3}X_{e1})}{Y^|_{e2, e3} - X_{e2}X_{e3}}}, \\
    q_{e1} &= X_{e1} - s_{e1}\sqrt{\frac{(Y^|_{e1,e2} - X_{e1}X_{e2})(Y^|_{e3, e1} - X_{e3}X_{e1})}{Y^|_{e2, e3} - X_{e2}X_{e3}}}.
\end{align*}
Furthermore, for all $e1, e2 \in \{ua, ub, uc \}$ and $e1\neq e2$, 
$s_{e1}s_{e2} = \mathsf{sgn}(Y_{e1,e2} - X_{e1}X_{e2})$.
\end{lemma}

\paragraph{\textsc{LearnLine}:} Similar to star vertex, we can use the expression in Lemma~\ref{lem:line} to design an algorithm \textsc{LearnLine}, which takes as input a line vertex $u$, the set of the edges of $E_1 \cup E_2$, and the limit of the estimators $X_{ua}$, $X_{ub}$, $X_{bc}$,  $Y^|_{ua, ub}$, $Y^|_{ub, bc}$, $Z^|_{ua, ub, bc}$ for $a$, $b$ and $c$ as in Figure \ref{fig:line}, and returns the weights of the edges $(u,a), (u,b)$ and $(b,c)$ in both mixtures (see Supplementary Materials)%, Algorithm \ref{app:learnLine}). 

\paragraph{\textsc{Learn2Nodes}}
Our main algorithm is initialized by learning weights associated with edges connected to two nodes using subroutine \textsc{Learn2Nodes}.
As this algorithm is very similar in spirit to our general algorithm, we leave the details to the Supplementary Materials. 

\begin{comment}
We first find the vertex with highest degree in the graph, say $u_{0}$. Further, let $v_{0}$ be the neighbor of $u_{0}$ with the maximum degree among the other neighbors. 

If $deg(u_{0}) \geq 3$ and $deg(v_{0})\geq 3$ then we can use \textsc{LearnStar} on vertices $u_{0}$ and $v_{0}$ to learn the weights of their adjacent edges. If $deg(u_{0}) \geq 3$ and $deg(v_{0}) =  2$, we need to use \textsc{LearnStar} on vertex $u_{0}$ and \textsc{LearnLine} on $v_{0}$. Note $deg(v_{0}) =  1$ is already learned with $u_0$.

Otherwise, we are left with a path graph.  In particular, we have either  $deg(u_{0}) = 2$ and $deg(v_{0}) = 2$, or   $deg(u_{0}) = 2$ and $deg(v_{0}) = 1$. In both cases we are able to recover the edges adjacent to $u_{0}$ and $v_{0}$  using \textsc{LearnLine}. However, due to the different orientation of the paths in these two cases, we need to use the \textsc{LearnLine} subroutine differently. The algorithm  \textsc{Learn2Nodes} and its proof of correctness can be found in the supplementary material.
\end{comment}

\subsection{Correctness of Algorithm~\ref{algo:generalGraph}}

To prove the correctness of the main algorithm, we show the following invariant:

\begin{lemma}\label{lem:invariant}
	At any point in the algorithm, the entire neighborhood of any node of $S$ has been learned and recorded in $W$:
	$$\forall u \in S, ~\forall v \in V,  ~ (u,v) \in E \implies (u,v) \in W.$$
\end{lemma}

\vspace{-0.4cm}
\begin{proof}
	We prove the above by induction on the iteration of the while loop. Due to the correctness of \textsc{Learn2Nodes} (proved in Supplementary material), after calling this function, $W$ contains all edges adjacent to the two vertices in $S$ . Hence the base case is true. Let us assume that after $k$ iterations of the loop, the induction hypothesis holds.
	
	We consider three cases in the $(k+1)$-th iteration:\\
		$\bullet$ $\mathbf{deg(v) \geq 3}$:  We recover all edges adjacent to the star vertex $v$ by using \textsc{LearnStar} (correct due to Lemma~\ref{lem:star}). Sign consistency is ensured using edge $(u,v)\in W$ since $u \in S$. \\
		$\bullet$ $\mathbf{deg(v) = 2}$: There exists $w \in S$ such that $(u,w) \in E$ since $|S| \ge 2$ and is connected. Since $deg(v) = 2$, there exists $t \ne u$ such that $(t, v) \in E$. Now if $t \in S$ then $(t,v) \in W$ and we are done. If $t \not \in S$ then $v$ is a line vertex for $t-v-u-w$. By using \textsc{LearnLine}  we recover all edges on the line (correct due to Lemma~\ref{lem:line}). Sign consistency is ensured through edge $(v,u)$.\\
		$\bullet$ $\mathbf{deg(v) = 1}$: Since $u \in S$, we have $(u,v) \in W$, so we are done.

	Thus by induction, after every iteration of the for loop, the invariant is maintained.
\end{proof}

\begin{theorem}\label{thm:algocorrect}
Suppose Condition \ref{ass:necessary} and \ref{ass:main} are true, Algorithm \ref{algo:generalGraph} learns the edge weights of the two balanced mixtures in the setting of infinite samples.
\end{theorem}

\vspace{-0.4cm}
\begin{proof}
Since at every iteration, the size of $S$ increases by 1, after at most $|V|$ iterations, we have $S = V$. Using Lemma~\ref{lem:invariant} we also have $W = E_1 \cup E_2$.
\end{proof}

\subsection{Finite Sample Complexity}
In this section, we investigate the error in estimating the quantities $X_{ui}$, $Y_{ui, uj}$ for $i,j \in \{a,b,c\}$ in case of a star vertex,  and  $X_{e1}$, $Y^|_{e1, e2}$ and $Z_{ua, ub, bc}$ for $e1,e2 \in \{ua,ub,bc\}$ for line vertex, using finite number of cascades. We further investigate the effect of the error in these quantities on the accuracy of the recovered weights. 

We use a simple count based estimator. Specifically, for events $\mathcal{E}_1$ and $\mathcal{E}_2$, we estimate the probability $\Pr(\mathcal{E}_1 | \mathcal{E}_2) = \frac{\sum_{m = 1}^M \mathbb{1}_{\mathcal{E}_1 \cap \mathcal{E}_2}}{\sum_{m = 1}^M \mathbb{1}_{\mathcal{E}_2}}$. As a concrete example, we have the estimator for $X_{ua}$ as $\hat{X}_{ua} := \frac{\sum_{m = 1}^M \mathbb{1}_{u \rightarrow a, u \in I_0^m}}{\sum_{m = 1}^M \mathbb{1}_{u \in I_0^m}}$. Here $I_0^m$ denotes the source of the $m$-th cascade and $u\rightarrow a$ denotes that $u$ infects $a$. We can argue using the law of large number and Slutsky's Lemma, that the above approach provides us with balanced estimators. 

We first establish high probability error bounds for the base estimators with finite number of cascades, for both the star vertex and the line vertex. Finally, using the above guarantees we provide our main sample complexity result for the balanced mixture problem.
See Supplemental material for proofs.

\begin{theorem}\label{thm:sampleupper}
	Suppose Condition \ref{ass:necessary} and \ref{ass:main} hold. With ${M  = \OO\left( \tfrac{1}{p_{min}^{6}\Delta^4} \frac{N}{\epsilon^2}\log\left( \frac{N}{\delta}\right) \right)}$ samples, Algorithm~\ref{algo:generalGraph} learns the edge weights of a balanced mixture on two graphs within precision $\epsilon$ with probability at least $1- \delta$.
\end{theorem}

\section{Extensions}
\subsection{Extension to Directed Graphs}
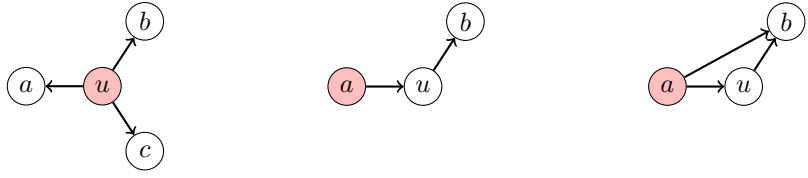
\begin{figure}
	\centering
	\begin{subfigure}[t]{0.28\textwidth}
		\centering
		\begin{tikzpicture}
		\node[main node, fill=white] (0) {$a$};
		\node[main node, fill=pink] (1) [right = 0.5cm of 0] {$u$};
		\node[main node, fill=white] (2) [above right = 0.5cm and 0.2cm of 1]  {$b$};
		\node[main node, fill=white] (3) [below right = 0.5cm and 0.2cm of 1] {$c$};
		
		\path[draw,thick,->]
		(1) edge node {} (0)
		(1) edge node {} (2)
		(1) edge node {} (3);
		\end{tikzpicture}
		\caption{A star vertex $u$ for a directed graph.}\label{fig:stardirected}
	\end{subfigure}\quad
	\begin{subfigure}[t]{0.28\textwidth}
	\centering
	\begin{tikzpicture}
	\node[main node, fill=pink] (0) {$a$};
	\node[main node, fill=white] (1) [right = 0.5cm of 0] {$u$};
	\node[main node, fill=white] (2) [above right = 0.5cm and 0.2cm of 1]  {$b$};
	\node[empty node] (3) [below right = 0.5cm and 0.2cm of 1] {};
	
	\path[draw,thick, ->]
	(0) edge node {} (1)
	(1) edge node {} (2);
	\end{tikzpicture}
	\caption{First structure to ensure sign consistency.}\label{fig:fork}
	\end{subfigure}\quad
	\begin{subfigure}[t]{0.28\textwidth}
	\centering
	\begin{tikzpicture}
	\node[main node, fill=pink] (0) {$a$};
	\node[main node, fill=white] (1) [right = 0.5cm of 0] {$u$};
	\node[main node, fill=white] (2) [above right = 0.5cm and 0.2cm of 1]  {$b$};
	\node[empty node] (3) [below right = 0.5cm and 0.2cm of 1] {};
	
	\path[draw,thick, ->]
	(0) edge node {} (1)
	(1) edge node {} (2)
	(0) edge node {} (2);
	\end{tikzpicture}
	\caption{Second structure to ensure sign consistency.}\label{fig:fork}
\end{subfigure}
\caption{Structures for directed graphs of minimum out-degree three.}\label{fig:directed}
\end{figure}

We notice that in the case of directed graphs of minimum out-degree three, we can simply use the star structure to learn all the directed edges. This however would not be enough to ensure sign consistency; we therefore need to also use both the structure in Figure \ref{fig:directed}. The algorithm is then very similar to Algorithm \ref{algo:generalGraph},  precise details are given in Supplementary Material.

\subsection{Extension to Unbalanced/Unknown Priors}
While previous results only considered balanced mixtures, $i.e.$ with parameter $\alpha = 0.5$, we focus here on unbalanced mixtures ($\alpha \neq 0.5$ known) and on mixtures of unknown priors ($\alpha$ unknown). 

We first note that, the main algorithm for recovering the graph does not depend on the prior once the correct \textsc{LearnStar} and \textsc{LearnLine} primitives are provided. Therefore, we focus here on designing correct  \textsc{LearnStar} and \textsc{LearnLine} primitives. 

\paragraph{Unbalanced Mixtures} We can easily extend our techniques for star vertices in the case of unbalanced mixtures, as we can get similar simplification as in Equation \ref{eq:star}. Specifically, we have for all $i\neq j \in \{a,b,c\}$:
$$Y_{ui,uj} - X_{ui} X_{uj} = \alpha(1-\alpha) (p_{ui} - q_{ui})(p_{uj} - q_{uj}).$$
However, equation \ref{eq:line} does not extend easily. Therefore, we get an increased dependency in $\frac{!}{\Delta}$ in the general case. The results are summarized in the theorem below, details can be found in Supplementary Material.

\begin{theorem}\label{thm:mainunbalanced}
	Suppose Conditions \ref{ass:necessary} and \ref{ass:main} are true. Then there exists an algorithm that runs on epidemic cascades over an \textbf{unbalanced} mixture of two undirected, weighted graphs $G_1 = (V, E_1)$ and $G_2 = (V, E_2)$,  with $|V| = N$, and recovers the edge weights corresponding to each graph up to precision $\epsilon$ in time $O(N^2)$ and sample complexity:
	\begin{itemize}
		\item $O\left(\frac{N\log N}{\epsilon^2} \mathrm{poly}(\tfrac{1}{\Delta}) \mathrm{poly}(\tfrac{1}{\min(\alpha, 1-\alpha)}) \right)$ in general.
		\item $O\left(\frac{N\log N}{\epsilon^2\Delta^2} \mathrm{poly}(\tfrac{1}{\min(\alpha, 1-\alpha)}) \right)$ for graph of minimum degree three.
	\end{itemize} 
\end{theorem}

\paragraph{Mixtures of Unknown Priors} If the graph has at least one star vertex, we can learn the entire mixture by learning the parameter $\alpha$ from this node, and using the results from above to learn the rest of the graph once $\alpha$ has been recovered. Details can be found in Supplementary Material.

\section{Conclusion}
We tackle the problem of learning the edge weights of a mixture of undirected graphs from epidemic cascades. Our algorithm is optimal (up to log factors) in term of sample complexity. %If the guarantees provided allow us to learn every weight of every edge within precision $\epsilon$, our algorithm achieves faster convergence on edges which are adjacent to edges for which the weights are very different in both graphs. 

As mentioned before, adding edges and nodes cannot make our algorithm fail, even if those nodes violate Conditions \ref{ass:necessary} or \ref{ass:main}. Indeed, if these Conditions are violated, a simple modification of our algorithm can recover the biggest connected subgraph which satisfies these conditions (simply restart the algorithm on new nodes whose neighborhood has not been learned yet). This implies that if the original graph satisfies Condition \ref{ass:necessary} and \ref{ass:main}, our algorithm is robust to any adversarial attacks which consist of adding nodes or edges.

Due to its local nature, our algorithm can be parallelized easily by computing the necessary \textsc{LearnStar}, and \textsc{LearnLine} sub routine calls in parallel. Some care should be given to ensure sign consistency. 

Our algorithm can be used to recover directed subgraphs of minimum out-degree three, which is powerful enough for most applications on social network. However, the techniques shown in this paper are not sufficient to learn general mixtures of directed graphs, or even to identify which class of mixture directed graphs can be learned. This is left fot future work.

\newpage
\bibliographystyle{plain}
\bibliography{all,constantine}

\newpage
\appendix

\appendix
%\aistatsaddress{UT Austin \And UT Austin \And UT Austin \And UT Austin}
\section{Necessary Conditions}\label{app:counterExamples}
\subsection{We need at least three edges}
Let $G=(V, E_1 \cup E_2)$ be the union of the graphs from both mixtures.  In this subsection, we prove it is impossible to learn the weights of $E_1$ and $E_2$ if $G$ has less than three edges:

\paragraph{One edge:} For a graph on two nodes, we have already seen that the cascade distribution are identical if $p_{12} = \beta = 1 - q_{12}$, for any value of $\beta$, which proves the problem is not solvable.

\paragraph{Two edges:} When we have two nodes and two edges, we can without loss of generality assume that node 1 is connected to node 2 and node 3. Then, if:
\begin{itemize}
	\item $p_{12} = \beta$
	\item $q_{12} = 1 - \beta$
	\item $p_{13} = \frac{ \frac{1}{2} - \frac{\beta}{2} + \frac{1}{4}}{\frac{1}{2} - \beta}$
	\item $q_{13} = \frac{\frac{1}{4} - \frac{\beta}{2}}{\frac{1}{2} - \beta}$
\end{itemize}
The cascade distribution is identical for any value of $\beta < \frac{1}{2}$. By simple calculations, we can show the following,
\begin{itemize}
	\item Fraction of cascades with only node 1 infected: $\frac{1}{12}$.
	\item Fraction of cascades with only node 2 infected: $\frac{1}{6}$.
	\item Fraction of cascades with only node 3 infected: $\frac{1}{6}$.
	\item Fraction of cascades where 3 infected 1, but 1 did not infect 2: $\frac{1}{12}$.
	\item Fraction of cascades where 3 infected 1, 1 infected 2: $\frac{1}{12}$.
	\item Fraction of cascades where 1 infected 3, but 1 did not infect 2: $\frac{1}{12}$.
	\item Fraction of cascades where 1 infected 2, but 1 did not infect 3: $\frac{1}{12}$.
	\item Fraction of cascades where 1 infected 3 and 2: $\frac{1}{12}$.
	\item Fraction of cascades where 2 infected 1, but 1 did not infect 3: $\frac{1}{12}$.
	\item Fraction of cascades where 2 infected 1, then 1 infected 3: $\frac{1}{12}$.
\end{itemize}
Since the distribution of cascades is the same for any value of $\beta < \frac{1}{2}$, the problem is not solvable.

\subsection{We need $\Delta$-separation}
%We now discuss the significance of the assumptions in the above theorem,  in particular we explain if they come from impossibility results or restrictions in our algorithm.
% $\bullet$ Condition \ref{ass:main} is necessary for the existence of sample efficient algorithms. Specifically, we show that there exist (many) graphs where Condition \ref{ass:main} is violated, and for which the sample complexity is at least exponential in the size of the graph. 
Separability is necessary for the existence of sample efficient algorithms. Specifically, we show that there exist (many) graphs where separability is violated, and for which the sample complexity is exponential in the size of the graph.

Indeed, consider a graph $G$ composed of two subgraphs $A$ and $B$, connected by a path $P$ of length $d$. Suppose the path has the same weight in both mixtures, and for the edges $e\in P$, $\max_{e\in P} p_e < 1$. Similar to the disconnected graph, we write $A_i = A \cap E_i$, and  $B_i = B \cap E_i$. To learn the graph completely we need to differentiate between the mixture on $E_1$ and $E_2$, and the mixture on $E'_1 = A_1 \cup P \cup B_2$ and $E'_2 = A_2 \cup P \cup B_1$. 

The path $P$ is not informative in the above differentiation as both the mixture in the path have same weights. Therefore, we need at least one cascade covering {\em at least one edge in $A$ and one edge in $B$}. Since $P$ is of length $d$, this happens with probability at most $e^{-\Omega(d)}$. To see such a cascade, we need at least $e^{\Omega(d)}$ cascades in expectation. Therefore, setting $d = c N$, for some constant $c>0$, we prove that exponential number of samples are necessary for any algorithm to recover the graph if the $\Delta$-separated Condition is violated.

\subsection{Dealing with mixtures which are not $\Delta$-separated} \label{app:nonDistinct}
In this section, we show how to detect and deduce the weights of edges which have the same weight across both component of the mixture. We assume  both $G_1$ and $G_2$ follow Conditions \ref{ass:necessary} and  \ref{ass:main} if we remove all non-distinct edges, and in particular  remain connected.

Suppose there exists an edge $(i,j)$ in the graph, such that $p_{ij} = q_{ij} > 0$. Then in particular, there exists another edge connecting $i$ to the rest of the graph $G_1$ through node $k$, such that $p_{ik} \neq q_{ik}$. Then:

\begin{lemma}
	Suppose $G_1$ and $G_2$ follow assumption \ref{ass:main} after removing all non-distinct edges. We can detect and learn the weights of non-distinct edges the following way:
	
	If  $X_{ij} > 0$, and $\forall k \in V, ~X_{ik} >0 \implies Y_{ik, ij} - X_{ik}X_{ij} = 0$, then $p_{ij} = q_{ij} = X_{ij}$.
\end{lemma}
\begin{proof}
Since $G_1$ is connected on three nodes or more even when removing edge $(i,j)$, we know there exists a node $l$ such either $l$ is connected to either $i$ or $k$. Therefore, either $Y_{ik, il} - X_{ik}X_{il} > 0$ or $Y_{ki, kl} - X_{ki}X_{kl} > 0 $. In both these cases, we deduce $p_{ik} \neq q_{ik}$. This in turns allow us to detect that $p_{ij} = q_{ij}$. Once this edge is detected, it is very easy to deduce its weight, since $p_{ij} = X_{ij} = q_{ij}$ by definition.
\end{proof}

\section{Proofs for unbalanced mixtures}
\subsection{Estimators - proofs}
\begin{lemma}\label{app:line}
		Under Conditions \ref{ass:necessary} and  \ref{ass:main}, in the setting of infinite samples, the weights of the edges for a line structure are then given by: 
	\begin{align*}
	p_{ua} &= X_{ua} + s_{ua} \sqrt{\frac{(Y^|_{ua,ub} - X_{ua}X_{ub})R^| }{Y^|_{ub, bc} - X_{ua}X_{bc}}},~q_{ua} = X_{ua} - s_{ua} \sqrt{\frac{(Y^|_{ua,ub} - X_{ua}X_{ub})R^| }{Y^|_{ub, bc} - X_{ua}X_{bc}}}, \\
	p_{bc} &= X_{uc} + s_{bc} \sqrt{\frac{(Y^|_{ub, bc} - X_{ua}X_{bc})R^| }{Y^|_{ua,ub} - X_{uc}X_{ua}}},~q_{bc} = X_{uc} - s_{bc} \sqrt{\frac{(Y^|_{ub, bc} - X_{ua}X_{bc})R^| }{Y^|_{ua,ub} - X_{uc}X_{ua}}}\\
	p_{ub} &= X_{ub} + s_{ub} \sqrt{\frac{(Y^|_{uc,ua} - X_{uc}X_{ua})(Y^|_{ub, bc} - X_{ua}X_{bc})}{R^|}}, \\
	q_{ua} &= X_{ua} - s_{ua} \sqrt{\frac{(Y^|_{ua,ub} - X_{uc}X_{ua})(Y^|_{ub, bc} - X_{ua}X_{bc})}{R^|}}, \\
	\end{align*}
	where $R^| = X_{ua}X_{bc} + \frac{Z^|_{ua, ub, bc} - X_{ua}Y^|_{ub, bc} - X_{bc}Y^|_{ua,ub}}{X_{ub}}$, and for $s_{ua} \in \{-1,1\}$.
\end{lemma}
\begin{proof} 
	In this case, there is no edge between $u$ and $c$, which implies that $p_{uc} = q_{uc} = 0$.  Hence, we cannot use a variation of the equation above for finding the edges of a star structure without dividing by zero. Therefore, we need to use $Z^|_{ua, ub, bc}$. 
	Let $R^| = X_{ua}X_{bc} + \frac{Z^|_{ua, ub, bc} - X_{ua}Y^|_{ub, bc} - X_{bc}Y^|_{ua,ub}}{X_{ub}}$. We notice a remarkable simplification:
	\begin{align*}
	R^| &= X_{ua}X_{bc} + \frac{Z^|_{ua, ub, bc} - X_{ua}Y^|_{ub, bc} - X_{bc}Y^|_{ua,ub}}{X_{ub}} \\
	&= \frac{p_{ua} + q_{ua}}{2}\cdot\frac{p_{bc} + q_{bc}}{2} + \frac{\frac{p_{ua}p_{ub}p_{bc} + q_{ua}q_{ub}q_{bc}}{2} - \frac{p_{ua} + q_{ua}}{2} \cdot \frac{p_{ua}p_{bc} + q_{ua}q_{bc}}{2} - \frac{p_{bc} + q_{bc}}{2}\cdot\frac{p_{ua}p_{ua} + q_{ua}q_{ua}}{2}}{\frac{p_{ub} + q_{ub}}{2}} \\
	&= \frac{1}{4} (p_{ua}p_{bc} + p_{ua}q_{bc} + q_{ua}p_{bc} + q_{ua}q_{bc}) + \frac{2}{p_{ub} + q_{ub}}\left[\frac{p_{ua}p_{ub}p_{bc} + q_{ua}q_{ub}q_{bc}}{2}  \right.\\
	&\quad \left.- \frac{1}{4}(p_{ua}p_{ub}p_{bc} + q_{ua}p_{ub}p_{bc} + p_{ua}q_{ub}q_{bc} + q_{ua}q_{ub}q_{bc}) \right.\\
	&\quad \left. -\frac{1}{4}(p_{ua}p_{ub}p_{bc} + p_{ua}p_{ub}q_{bc} +  q_{ua}q_{ub}p_{bc} + q_{ua}q_{ub}q_{bc})\right] \\
	&= \frac{1}{4} (p_{ua}p_{bc} + p_{ua}q_{bc} + q_{ua}p_{bc} + q_{ua}q_{bc}) - \frac{1}{2(p_{ub} + q_{ub})}\left[q_{ua}p_{ub}p_{bc} + p_{ua}q_{ub}q_{bc} + p_{ua}p_{ub}q_{bc} +  q_{ua}q_{ub}q_{bc} \right] \\
	&= \frac{1}{4} (p_{ua}p_{bc} + q_{ua}p_{bc} + q_{ua}q_{bc}+ p_{ua}q_{bc} )  - \frac{1}{2(p_{ub} + q_{ub})}\left[(p_{ub} + q_{ub})( q_{ua}p_{bc} + p_{ua}q_{bc}) \right] \\
	&= \frac{1}{4}(p_{ua}p_{bc} + q_{ua}q_{bc} - p_{ua}q_{bc} - q_{ua}p_{bc}) \\
	&= \frac{1}{4}(p_{ua} - q_{ua})(p_{bc} - q_{bc})
	\end{align*}
	
	We can then use the same proof techniques as in Lemma \ref{lem:star}, and finally obtain:
	\begin{align*}
	|p_{ua} - q_{ua}| &= \sqrt{\frac{(Y^|_{ua,ub} - X_{ua}X_{ua})\left(X_{ua}X_{bc} + \frac{Z^|_{ua, ub, bc} - X_{ua}Y^|_{ub, bc} - X_{bc}Y^|_{ua,ub}}{X_{ub}}\right) }{Y^|_{ub, bc} - X_{ua}X_{bc}}}.
	% , \\
	% |p_{ub} - q_{ub}| &= \sqrt{\frac{(Y^|_{ua,ub} - X_{ua}X_{ub})(Y^|_{ub, bc} - X_{ua}X_{bc})}{\left(X_{ua}X_{bc} + \frac{Z^|_{ua, ub, bc} - X_{ua}Y^|_{ub, bc} - X_{bc}Y^|_{ua,ub}}{X_{ub}}\right)}}, \\
	% |p_{bc} - q_{bc}| &= \sqrt{\frac{(Y^|_{ub, bc} - X_{ua}X_{bc})\left(X_{ua}X_{bc} + \frac{Z^|_{ua, ub, bc} - X_{ua}Y^|_{ub, bc} - X_{bc}Y^|_{ua,ub}}{X_{ub}}\right) }{Y^|_{ua,ub} - X_{ua}X_{ua}}}.
	\end{align*}
	This gives us the required result.
\end{proof}

\subsection{Resolving Sign Ambiguity across Base Estimators}
The following lemma handles the sign ambiguity ($s_{ua}$) introduced above.
\begin{lemma}\label{lem:sign}
	Suppose Condition \ref{ass:necessary} and \ref{ass:main} are true, in the setting of infinite samples, for edges $(u,a), (u,b)$ with $a\ne b$ for any vertex $u$ with degree $\ge 2$, the sign pattern $s_{ua}, s_{ub}$ satisfy the following relation.
	\[
	s_{ua}s_{ub} = \mathsf{sgn}(Y_{ua,ub} - X_{ua}X_{ub}).
	\]
\end{lemma}
\begin{proof}
	From previous analysis, we have $\mathsf{sgn}(p_{ua} - q_{ua}) = s_{ua}$. Therefore:
	\begin{align*}
	\mathsf{sgn}(Y_{ua,ub} - X_{ua}X_{ub}) &= \mathsf{sgn}\left(\frac{(p_{ua} - q_{ua})(p_{ub} - q_{ub})}{4}\right) \\
	&=  s_{ua}s_{ub}.
	\end{align*}
\end{proof}
Thus fixing sign of one edge gives us the signs of all the other edges adjacent to a star vertex. A similar relationship can be established among the edges of a line vertex, using $ \mathsf{sgn}\left(X_{ua}X_{bc} + \frac{Z^|_{ua, ub, bc} - X_{ua}Y^|_{ub, bc} - X_{bc}Y^|_{ua,ub}}{X_{ub}}\right) $.

\subsection{Main algorithm - proofs} \label{app:algo}
Here we will present in detail the sub-routines required by our algorithm and the essential lemmas needed for our main proof.
\paragraph{LearnEdges} This procedure detects the edges in the underlying graph using the estimate $X_{uv}$.
\begin{algorithm}[H] 
	\caption{$\textsc{LearnEdges}$}\label{app:learnEdges}
	\hspace*{\algorithmicindent} \textbf{Input} Vertex set $V$\\
	\hspace*{\algorithmicindent} \textbf{Output} Edges of the graph
	\begin{algorithmic}[1]
		\State Set $E \leftarrow \emptyset$
		\For {$u < v \in V$}
		\State Compute $\hat{X}_{uv}$
		\If {$\hat{X}_{uv} \ge \epsilon$}
		\State $E \leftarrow E \cup \{(u,v)\}$
		\EndIf
		\EndFor
		\State Return $E$
	\end{algorithmic}
\end{algorithm}
\begin{claim}\label{cl:Xab}
	$\textsc{LearnEdges}(V)$ outputs $E$ such that $E = E_1 \cup E_2$.
\end{claim}
\begin{proof}
	For each pair of nodes $u,v \in V$, if $(u,v) \in E_1 \cup E_2$ then $X_{uv} \ne 0$ since $X_{uv} = 0$ if and only if $p_{uv} = q_{uv} = 0$, which is equivalent to the edge $(u, v)$ not belonging in the mixture.
\end{proof}
\paragraph{LearnStar} This procedure returns the weights of the outgoing edges of a star vertex using the star primitive discussed before.
\begin{algorithm}[H] 
	\caption{$\textsc{LearnStar}$}\label{app:learnStar}
	\hspace*{\algorithmicindent} \textbf{Input} Star vertex $u \in V$, edge set $E$, weights $W$\\
	\hspace*{\algorithmicindent} \textbf{Output} Weights of edges adjacent to $u$
	\begin{algorithmic}[1]
		\State Use star primitive with star vertex $u$ and learn all adjacent edges weights $W^*$.
		\If {$W = \emptyset$}
		\State Fix sign of any edge and ensure sign consistency.
		\Else
		\State Set $v \in V$ such that $(u,v) \in W$.
		% \State Set $N = \{w \in V | (u,w) \in E, (u,w) \not \in W\}$
		\State Use $s_{uv}$ to remove sign ambiguity 
		\EndIf
		\State Return $W^*$.
	\end{algorithmic}
\end{algorithm}
% The $\textsc{LearnStar}(u,a, S, W)$ runs as follows:
% \begin{enumerate}
%     \item Using $a$ as the star node in the star structure, we use estimates to recover the edge weights of all edges adjacent to $a$ using $u$ as a fixed node.
%     \item Fix sign consistencies using $s_{ua}$.
% \end{enumerate}
\begin{lemma}\label{lem:learnstar}
	If $deg(u) \ge 3$, $\textsc{LearnStar}(u, S, W)$ recovers $p_{ua}, q_{ua}$ for all $a$ such that $(u,a) \in E$. 
\end{lemma}
\begin{proof}
	The proof follows from using Lemma \ref{lem:star} on star vertex $u$ (degree of $u \ge 3$) and using Lemma \ref{lem:sign} to resolve sign ambiguity through fixing an edge or $s_{uv}$ ($(u,v) \in W$ hence know sign). 
\end{proof}

\paragraph{LearnLine} This procedure returns the weights of the edges of a line $a-b-c-d$ rooted at vertex $b$ of degree 2 using the line primitive discussed before.
\begin{algorithm}[H] 
	\caption{$\textsc{LearnLine}$}\label{app:learnLine}
	\hspace*{\algorithmicindent} \textbf{Input} Line $a-b-c-d$ with $deg(b) = 2$, edge set $E$, weights $W$\\
	\hspace*{\algorithmicindent} \textbf{Output} Weights of edges $(a,b), (b,c), (c,d)$
	\begin{algorithmic}[1]
		\State Use line primitive on $a-b-c-d$ rooted at $b$ and learn all edges weights $W^|$.
		\If {$W = \emptyset$}
		\State Fix sign of any edge and ensure sign consistency.
		\Else
		\State Find edge $e \in \{(a,b), (b,c), (c,d)\}$ such that $e \in W$.
		% \State Set $N = \{w \in V | (u,w) \in E, (u,w) \not \in W\}$
		\State Use $s_{e}$ to remove sign ambiguity.
		\EndIf
		\State Return $W^|$.
	\end{algorithmic}
\end{algorithm}
% The $\textsc{LearnLine}(u,a, S, W)$ runs as follows:
% \begin{enumerate}
%     \item Choose vertex $c \in S$ such that $(u,c) \in E$. Choose vertex $b\ne c, u$ such that $(a,b) \in E$. If there does not exist such $b$, return empty set.
%     \item Using line structure on $c-u-a-b$ rooted at $a$, use estimates to recover the edge weights of edge $(a,b)$.
%     \item Fix sign consistencies using $s_{ua}$.
% \end{enumerate}
\begin{lemma}\label{lem:learnline}
	If $deg(b) = 2$, $\textsc{LearnLine}(a, b, c, d, S, W)$ recovers $p_{ab}, q_{ab}, p_{bc}, q_{bc}, p{cd}, q_{cd}$. 
\end{lemma}
\begin{proof}
	The proof follows from using Lemma \ref{lem:line} on line $a-b-c-d$ rooted at vertex $b$ (degree of $b= 2$) and using Lemma \ref{lem:sign} to resolve sign ambiguity by fixing an edge or using $s_{e}$.
\end{proof}

\paragraph{Learn2Nodes} This procedure chooses a pair of connected vertices in our graph and outputs the weights of all outgoing edges of each of the two vertices. We initialize our algorithm using this procedure.
\begin{algorithm}[H] 
	\caption{$\textsc{Learn2Nodes}$}\label{app:learn2Nodes}
	\hspace*{\algorithmicindent} \textbf{Input} Vertex set $V$, Edge Set $E$\\
	\hspace*{\algorithmicindent} \textbf{Output} Set of 2 vertices $V$, Weight of all edges adjacent to the vertices $W$
	\begin{algorithmic}[1]
		\State $W = \emptyset$
		\State Set $u = \argmax_{a \in V} deg(a)$
		\State Set $v = \argmin_{a \in V, (u,a) \in E} deg(a)$
		\If {$deg(u) \ge 3$}
		\State $W \leftarrow \textsc{LearnStar}(u, E, W)$
		\If {$deg(v) = 3$}
		\State $W \leftarrow W \cup \textsc{LearnStar}(v, E, W)$
		\ElsIf {$deg(v) = 2$}
		\State Let $t \in V$ be such that $(t, v) \in E$ and $t \ne u$
		\State Let $w \in V$ be such that $(w, u) \in E$ and $w \ne v,t$
		\If {$v = t$}
		\State $W \leftarrow W \cup \textsc{LearnLine}(t, v, u, w, W)$
		\EndIf
		\EndIf
		\Else
		\State $w$ be such that $(w, u) \in E$ and $w \ne v$
		\If {$deg(v) = 2$}
		\State Let $t \in V$ be such that $(t, v) \in E$ and $t \ne u$
		\State $W \leftarrow \textsc{LearnLine}(w, u, v, t, W)$
		\Else
		\State Let $t \in V$ be such that $(t, w) \in E$ and $t \ne w$
		\State $W \leftarrow \textsc{LearnLine}(v, u, w, t, W)$
		\EndIf
		\EndIf
		\State Return $(u,v), W$
	\end{algorithmic}
\end{algorithm}
\begin{lemma} \label{lem:learn2nodes}
	Under Conditions \ref{ass:necessary} and \ref{ass:main}, $\textsc{Learn2Nodes}(V)$ outputs two connected nodes $(u,v)$ and weights of all edges adjacent to $u,v$.
\end{lemma}
\begin{proof}
	We will break the proof down into cases based on the degree of chosen vertices $u,v$ as follows,
	\begin{itemize}
		\item $deg(u) \ge 3$: By Lemma \ref{lem:learnstar}, we can recover all the edges of $u$ and fix a sign.
		\begin{itemize}
			\item $deg(v) \ge 3$: By Lemma \ref{lem:learnstar}, we can recover all the edges of $v$ and ensure sign consistency by using the edge $(u,v)$.
			\item $deg(v) = 2$: Since $deg(v) = 2$, there exists a vertex $t \ne u$ such that $(t,u) \in E$. Since $deg(u) \ge 3$, there must exist $w \ne t, u$ such that $(u,w) \in E$. Now we have line primitive $t-v-u-w$  with $deg(v) = 2$ and Lemma \ref{lem:learnline} guarantees recovery of the edge weights.
			\item $deg(v) = 1$, then we already know all the edges adjacent to $v$.
		\end{itemize}
		\item $deg(u) = 2, deg(v) = 2$: Since the max degree of the graph is 2 and it is connected then it can either be a line or a cycle. There are at least 4 nodes in the graph, thus there exist $w \ne v$ such that $(w, u) \in E$ and $t \ne u, w$ such that $(v,t) \in E$. This gives a path $w-u-v-t$ with $deg(u) = 2$ and Lemma \ref{lem:learnline} guarantees recovery of all edges.
		\item $deg(u) = 2, deg(v) = 1$: As in the previous case, the underlying graph is a line. Therefore there exist path $v-u-w-t$ and we can similarly apply Lemma \ref{lem:learnline} to guarantee recovery of all edges. 
	\end{itemize}
\end{proof}

\subsection{Finite sample complexity - proofs} \label{app:finiteSamples}
In this section, we provide explicit proof for the sample complexity of our algorithm. To do so, we bound below the number of cascades starting on each node through Bernstein inequality, and use this number to obtain concentration of all the estimators.

\begin{definition}
	Among $M$ cascades, let $M_u$ be the number of times node $u$ is the source.
\end{definition}

\begin{claim}\label{cl:Mu}
	With $M$ samples, every node is the source of the infection at least $\frac{M}{2N}$ times with probability at least $1- e^{-\frac{3M}{26N}}$.
\end{claim}
\proof  Among $M$ cascade, the expectation of $M_u$ is $\frac{M}{N}$, since the source is chosen uniformly at random among the $N$ vertices of $V$. Since $M_u$ can be seen as the sum of Bernoulli variable of parameter $\frac{1}{N}$, we can use Bernstein's inequality to bound it below:
\begin{align*}
\Pr(M_u < \frac{M}{2N}) &= \Pr\left(\frac{M}{N} - M_u > \frac{M}{2N}\right)\\
&\leq e^{ - \frac{\left(\frac{M}{2N}\right)^2}{2M\frac{1}{N}(1- \frac{1}{N}) + \frac{1}{3}\frac{M}{2N}}} \\
&\leq e^{-\frac{3M}{26N}}.
\end{align*}

\begin{claim}\label{cl:epsilon1}
	Let $u$ either be a star vertex, with neighbors $a, b$ and $c$, or be part of a line structure rooted in $u$, with neighbors $a, b$, and $c$ neighbor of $b$. Suppose $M_u \geq \frac{M}{2N}$. Then with $M = \frac{N}{\epsilon^2} \log\left(\frac{12N^2}{\delta}\right)$ samples, with probability at least $1- \frac{\delta}{6N^2}$, we can guarantee any of the following:
	\begin{enumerate}
		\item $\forall r \in {a,b,c}, ~ \left|\hat{X}_{ur} - X_{ur}\right| \le  \epsilon_1 $.
		\item $\forall r \neq s \in \{a,b,c\}, ~|\hat{Y}^*_{ur,us} - \hat{Y}^*_{ur,us}| \le  \epsilon_1$.
		\item $|\hat{Y}^|_{ua,ub} - \hat{Y}^|_{ua,ub}| \le  \epsilon_1$ and $|\hat{Y}^|_{ua,ab} - \hat{Y}^|_{ua,ab}| \le  \epsilon_1$.
		\item $|\hat{Z}^|_{ua,ub,bc} - Z^|_{ua,ub,bc}| \le \epsilon_1$.
	\end{enumerate}
\end{claim}
\proof By Hoeffding's inequality:
\begin{align*}
\Pr( |\hat{X}_{ur} - X_{ur}| >  \epsilon_1) &= \Pr\left( \left|\sum_{m =1}^{M_u} \mathbb{1}_{\{u \rightarrow r~|~ u \in I_0\}} - M_u\cdot X_{ur} \right| >  M_{u}\cdot \epsilon_1 \right) \\
&\leq \Pr\left( \left|\sum_{m =1}^{\frac{M}{2N}} \mathbb{1}_{\{u \rightarrow r~|~ u \in I_0\}} - \frac{M}{2N}\cdot X_{ur} \right| >  \frac{M}{2N} \cdot \epsilon_1 \right) \\
&\leq \displaystyle 2e^{-2\frac{M}{2N}\epsilon_1^2}.
\end{align*}
Therefore, the quantity above is smaller than $\frac{\delta}{6N^2}$ for $M \geq \frac{N}{\epsilon_1^2} \log\left(\frac{12N^2}{\delta}\right)$. The proof is almost identical for the other quantities involved.

\begin{claim}\label{cl:epsilon2}
	If we can estimate $X_{ua}, Y^*_{ua,ub}, Y^|_{ua,ab}$ and $Z^|_{ua,ub,bc}$ within $\epsilon_1$, we can estimate $p_{ua}$ within precision $\epsilon = \frac{41}{p_{min}^3\Delta^2}\cdot \epsilon_1$.
\end{claim}
\proof 
If $u$ is of degree three or more, we use a star primitive to estimate it. Let $a, b$ and $c$ be three of its neigbors:
\begin{align*}
\hat{p}_{ua} &=  \hat{X}_{ua} + s_{ua} \sqrt{\frac{(\hat{Y}_{ua,ub} - \hat{X}_{ua}\hat{X}_{ub})(\hat{Y}_{ua,uc} - \hat{X}_{ua}\hat{X}_{uc})}{\hat{Y}_{ub,uc} - \hat{X}_{ub}\hat{X}_{uc}}} \\
&\leq X_{ua} + \epsilon_1 \\
&\quad+ s_{ua} \left(\frac{(Y_{ua,ub} - X_{ua}X_{ub} + s_{ua}\left(1 + X_{ua} + X_{ub}\right)\epsilon_1)(Y_{ua,uc} - X_{ua}X_{uc} + s_{ua}\left[1 + X_{ua} + X_{uc}\right)\epsilon_1) }{Y_{ub, bc} - X_{ub}X_{uc} - s_{ua}\left(1 + X_{ub} + X_{uc}\right)\epsilon_1)}\right]^\frac{1}{2} \\
&\leq X_{ua} + \epsilon_1 +  s_{ua} \sqrt{\frac{(Y_{ua,ub} - X_{ua}X_{ub})(Y_{ua,uc} - X_{ua}X_{uc})}{Y_{ub,uc} - X_{ub}X_{uc}}}\left(\frac{(1 + s_{ua}\frac{3\epsilon_1}{\frac{\Delta^2}{4}})^2}{1 - s_{ua}\frac{3\epsilon_1}{\frac{\Delta^2}{4}}} \right)^\frac{1}{2} \\
&\leq p_{ua} + \epsilon_1 + p_{ua} \left( \frac{12}{\Delta^2} + \frac{6}{\Delta^2}\right)\cdot \epsilon_1 + o(\epsilon_1) \\
&\leq p_{ua} + \frac{19}{\Delta^2}\cdot \epsilon_1 + o(\epsilon_1).
\end{align*}
Where we have used $Y_{ur,us} - X_{ur}X_{us} \geq \frac{\Delta^2}{4}$, $s_{ua}^2 = 1$, $p_{ua} \leq 1$, $1 \leq \frac{1}{\Delta^2}$. We then conclude by symmetry.

If $u$ is of degree two, we use a line primitive to estimate it:
\begin{align*}
\hat{p}_{ua} &= \hat{X}_{ua} + s_{ua} \sqrt{\frac{(\hat{Y}^|_{ua,ub} - \hat{X}_{ua}\hat{X}_{ub})\left(\hat{X}_{ua}\hat{X}_{bc} + \frac{\hat{Z}^|_{ua, ub, bc} - \hat{X}_{ua}\hat{Y}^|_{ub, bc} - \hat{X}_{bc}\hat{Y}^|_{ua,ub}}{\hat{X}_{ub}}\right) }{\hat{Y}^|_{ub, bc} - \hat{X}_{ua}\hat{X}_{bc}}}\\
&\leq X_{ua}  + \epsilon_1 + s_{ua} \sqrt{\frac{(Y^|_{ua,ub} - X_{ua}X_{ub} + 3s_{ua}\epsilon_1)\left(X_{ua}X_{bc} + 2\epsilon_1 + \frac{Z^|_{ua, ub, bc} - X_{ua}Y^|_{ub, bc} - X_{bc}Y^|_{ua,ub} + 5s_{ua}\epsilon_1 }{X_{ub} - s_{ua}\epsilon_1 } \right) }{Y^|_{ub, bc} - X_{ua}X_{bc} - 3s_{ua}\epsilon_1}}.
\end{align*}
As shown in the proof of Lemma \ref{lem:line}, we have:
\begin{align*}
Z^|_{ua, ub, bc} - X_{ua}Y^|_{ub, bc} - X_{bc}Y^|_{ua,ub}   &= \frac{1}{2}(p_{ub} + q_{ub})(q_{ua}p_{bc} + p_{ua}q_{bc}) \\
&\geq \frac{p^3_{min}}{2} \\
\left(X_{ua}X_{bc} + \frac{Z^|_{ua, ub, bc} - X_{ua}Y^|_{ub, bc} - X_{bc}Y^|_{ua,ub}}{X_{ub}} \right) &= \frac{1}{4}(p_{ua} - q_{ua})(p_{bc} - q_{bc}) \\
&\geq \frac{\Delta^2}{4}.
\end{align*}
Therefore:
\begin{align*}
\frac{Z^|_{ua, ub, bc} - X_{ua}Y^|_{ub, bc} - X_{bc}Y^|_{ua,ub} + 5s_{ua}\epsilon_1 }{X_{ub} - s_{ua}\epsilon_1 } &\leq      \frac{Z^|_{ua, ub, bc} - X_{ua}Y^|_{ub, bc} - X_{bc}Y^|_{ua,ub} }{X_{ub}}  \left[ \frac{1 + s_{ua}\frac{5\epsilon_1}{\frac{p^3_{min}}{2}}}{1 - s_{ua}\frac{\epsilon_1}{\frac{p_{min}}{2}}} \right] \\
&\leq  \frac{Z^|_{ua, ub, bc} - X_{ua}Y^|_{ub, bc} - X_{bc}Y^|_{ua,ub}}{X_{ub}} + s_{ua}\left( \frac{12}{p^3_{min}}\right) \epsilon_1 + o(\epsilon_1).
\end{align*}
We also have:
\begin{align*}
X_{ua}X_{bc} + 2\epsilon_1 + \frac{Z^|_{ua, ub, bc} - X_{ua}Y^|_{ub, bc} - X_{bc}Y^|_{ua,ub} + 5s_{ua}\epsilon_1 }{X_{ub} - s_{ua}\epsilon_1 } &\leq \left(X_{ua}X_{bc} + \frac{Z^|_{ua, ub, bc} - X_{ua}Y^|_{ub, bc} - X_{bc}Y^|_{ua,ub}}{X_{ub}} \right) 
\\&\quad \cdot \left(1 + s_{ua}\frac{\frac{14}{p_{min}^3}}{\frac{\Delta^2}{4}}\epsilon_1 \right) + o(\epsilon_1).
\end{align*}
Combining all the above inequalitites:
\begin{align*}
\hat{p}_{ua} &\leq X_{ua} + \epsilon_1 + s_{ua} \sqrt{\frac{(Y^|_{ua,ub} - X_{ua}X_{ub})\left(X_{ua}X_{bc} + \frac{Z^|_{ua, ub, bc} - X_{ua}Y^|_{ub, bc} - X_{bc}Y^|_{ua,ub}}{X_{ub}}\right) }{Y^|_{ub, bc} - X_{ua}X_{bc}}} \\
&\quad\cdot\left[ \frac{\left(1 + \frac{3\epsilon_1}{\frac{\Delta^2}{4}}\right)\left(1 + s_{ua}\frac{\frac{14}{p_{min}^3}}{\frac{\Delta^2}{4}}\epsilon_1 \right)}{1 - s_{ua}\frac{3\epsilon_1}{\frac{\Delta^2}{4}}} \right]^\frac{1}{2} \\
&\leq p_{ua} + \epsilon_1 + p_{ua}s_{ua}^2\left( \frac{6}{\Delta^2} + \frac{28}{p_{min}^3\Delta^2} + \frac{6}{\Delta^2} \right)\cdot \epsilon_1 + o(\epsilon_1) \\
&\leq p_{ua} + \frac{41}{p_{min}^3\Delta^2}\cdot \epsilon_1 + o(\epsilon_1).
\end{align*}
We can conclude by symmetry. 

Since $\frac{41}{p_{min}^3\Delta^2}\cdot \epsilon_1 \geq \frac{19}{\Delta^2}\cdot \epsilon_1$, we conclude that we can know $p_{ua}$ within precision $\epsilon = \frac{41}{p_{min}^3\Delta^2}\cdot \epsilon_1$ regardless of the degree of $u$.

\begin{theorem}
		Under Conditions \ref{ass:necessary} and  \ref{ass:main},, with probability $1- \delta$, with $M = N\cdot\frac{ 41^2}{p_{min}^{6}\Delta^4\cdot \epsilon^2} \log\left( \frac{12N^2}{\delta}\right) = \OO\left( \frac{N}{\epsilon^2}\log\left( \frac{N}{\delta}\right) \right)$ samples, we can learn all the edges of the mixture of the graphs within precision $\epsilon$.
\end{theorem}
\proof 
We pick $\epsilon = \frac{41}{p_{min}^3\Delta^2}\cdot \epsilon_1$. We use Claim \ref{cl:Mu} to bound the quantity $\Pr( M_u < \frac{M}{2N})$, and Claim \ref{cl:epsilon1} and \ref{cl:epsilon2} to bound $\Pr(|\hat{p}_{ua} - p_{ua}| > \frac{41}{p_{min}^3\Delta^2}\cdot \epsilon_1 | M_u \geq \frac{M}{2N})$.  For $(u,a)$ edge of the graph:
\begin{align*}
\Pr(|\hat{p}_{ua} - p_{ua}| > \epsilon) &\leq \Pr(|\hat{p}_{ua} - p_{ua}| > \epsilon | M_u < \frac{M}{2n})\cdot \Pr( M_u < \frac{M}{2N}) \\
&\quad+ \Pr(|\hat{p}_{ua} - p_{ua}| > \epsilon | M_u \geq \frac{M}{2n})\cdot \Pr( M_u \geq \frac{M}{2N}) \\
&\leq 1\cdot \displaystyle 2e^{-2\frac{M}{2N}}  + \Pr(|\hat{p}_{ua} - p_{ua}| > \epsilon | M_u \geq \frac{M}{2N})\cdot 1 \\
&\leq \frac{\delta}{12N^2} + \Pr(|\hat{p}_{ua} - p_{ua}| > \frac{41}{p_{min}^3\Delta^2}\cdot \epsilon_1 | M_u \geq \frac{M}{2N}) \\
&\leq \frac{\delta}{12N^2} + \frac{\delta}{12N^2} \\
&\leq \frac{\delta}{6N^2}.
\end{align*}
We conclude by union bound on the six estimators involved for all the pairs of nodes in the graph, for a total of at most $6N^2$ estimators.

\subsection{Complete graph on three nodes} \label{app:triangle}
In this section, we prove it is possible to recover the weights of a mixture on three nodes, as long as there are at least three edges in $E_1 \cup E_2$. Since no node is of degree 3, no node is a star vertex, and since there are less than four nodes, no node is a line vertex, and we can not use the techniques developped above for connected graphs on four vertices or more. However, we can still use very similar proofs techniques. Suppose the vertices of $V$ are $1$, 2 and 3.
\begin{definition}
	 We reuse the quantities defined for star vertices:
	\begin{itemize}
		\item For $i$, $j$ distinct in $\{1,2,3 \}, ~\hat{X}_{ij} =  \frac{\frac{1}{M}\sum_{m = 1}^M \mathbb{1}_{i \rightarrow j, i \in I_0^m}}{\frac{1}{M}\sum_{m = 1}^M \mathbb{1}_{i \in I_0^m}} \to_{M \to \infty} X_{ij} = \frac{p_{ij} + q_{ij}}{2}.$
		\item For $i$, $j$, $k$ distinct in $\{1,2,3 \}, ~Y_{ij,ik} = \frac{\frac{1}{M}\sum_{m = 1}^M \mathbb{1}_{i \rightarrow j, i \rightarrow k, i \in I_0^m}}{\frac{1}{M}\sum_{m = 1}^M \mathbb{1}_{u \in I_0^m}} \to_{M \to \infty} Y_{ij, ik}= \frac{p_{ij}p_{ik} + q_{ij}q_{ik}}{2}.$
	\end{itemize}
\end{definition}
Even though neither 1, 2 or 3 is a star vertex, we can write the same kind of system of equations as a star vertex would satisfy. In particular:

$$\frac{|p_{ij} - q_{ij}|}{2} = \sqrt{\frac{(Y_{ij,ik} - X_{ij}{ik})(Y_{ji, jk} - X_{ji}X_{jk})}{Y_{ki,kj} - X_{ki}X_{kj}}}. $$

Resolving the sign ambiguity as previoulsy (Lemma \ref{lem:sign}), this finally yields:

\begin{align*}
p_{ij} &= X_{ij} + s_{ij}\sqrt{\frac{(Y_{ij,ik} - X_{ij}{ik})(Y_{ji, jk} - X_{ji}X_{jk})}{Y_{ki,kj} - X_{ki}X_{kj}}}, \\
q_{ij} &= X_{ij} + s_{ij}\sqrt{\frac{(Y_{ij,ik} - X_{ij}{ik})(Y_{ji, jk} - X_{ji}X_{jk})}{Y_{ki,kj} - X_{ki}X_{kj}}}.
\end{align*}

\section{Lower Bounds} \label{app:lowerBound}
\subsection{Directed lower bound}
We consider the task of learning all the edges of any mixture of graphs up to precision $\epsilon < \Delta$. To do so, we have to be able to learn a mixture on a specific graph, which we present below.
\begin{figure}[H]
	\centering
	\begin{tikzpicture}
	\node[main node, fill=white] (4) {4};
	\node[main node, fill=white] (2) [right = 0.25cm of 4] {2};
	\node[main node, fill=white]  (3) [below = 0.25cm of 2] {3};
	\node[main node, fill=white]  (1) [above = 0.25cm of 2] {1};

	\node[main node, fill=white] (5) [above right = 2.5cm and 4cm of 2] {5}; 
	
	\node[main node, fill=white] (6) [below = 0.25cm of 5] {6};
	
	\node[main node, fill=white] (7) [below = 0.25cm of 6] {7};
	
	\node[main node, fill=white] (8) [below = 2cm of 7] {N}; 
	
	\path[dotted]
	(7) edge node {} (8);
	
	\path[->]
	(1) edge node {} (4)
	(1) edge node {} (2)
	(1) edge[bend right] node [right] {} (3)
	(3) edge[bend left] node [left] {} (1)
	(2) edge node {} (4)
	(2) edge node {} (3)
	(2) edge node {} (1)
	(3) edge node {} (4)
	(3) edge node {} (2)
	(4) edge node {} (3)
	(4) edge node {} (2)
	(4) edge node {} (1)
	(5) edge node {} (1)
	(5) edge node {} (3)
	(5) edge node {} (2)
	(6) edge node {} (1)
	(6) edge node {} (3)
	(6) edge node {} (2)
	(7) edge node {} (1)
	(7) edge node {} (3)
	(7) edge node {} (2)
	(8) edge node {} (1)
	(8) edge node {} (3)
	(8) edge node {} (2);
	\end{tikzpicture}
	\caption{Lower-bound directed graph}
\end{figure}
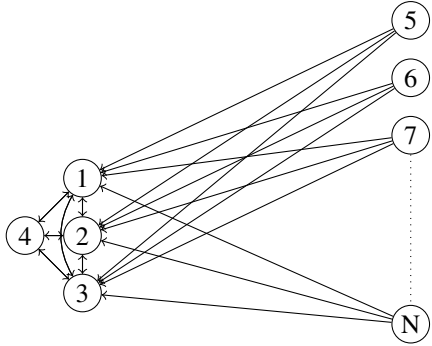

The example we focus on is the directed graph of min-degree 3, comprised of a clique on 4 nodes, which we call nodes 1 to 4, and $N-4$ other nodes with 3 directed edges to nodes 1, 2 and 3. All edges have weight $p$ in $E_1$, and $p+\Delta$ in $E_2$.

We define a \textit{valid sample} for edge $(i,j)$ as a cascade during which $i$ became infected when $j$ was not infected. Indeed, in this case, an infection could happen along edge $(i,j)$, and we can therefore gain information about the weight of this edge. We first state a general claim:

\begin{claim}\label{cl:validSamples}
	We need at least $\Omega(\frac{1}{\Delta^2})$ valid samples for edge $(i,j)$ to determine the weights of this edge in the mixture.
	\proof Using Sanov's theorem \cite{Sanov1961}, and writing the Kullback–Leibler divergence between $p$ and $q$ as $\mathcal{D}(p || q)$, we know we need at least $\Omega(\mathcal{D}(p || p + \Delta)) $ valid samples to determine whether the valid samples came from a random flip of probability $p$, or a random flip of probability $p + \Delta$, which is an easier task than computing both weights of the mixture. 
	
	Then, using standard Kullback–Leibler divergence bounds \cite{dragomir2000some}, we obtain $\mathcal{D}(p || p + \Delta) \geq \frac{1}{\Delta^2}$, which gives us the desired result.
\end{claim}

We now combine this with Coupon collector's result to obtain our lower bound. 

\begin{claim}
	We need at least $\Omega\left(N \log(N) + \frac{N \log\log(N)}{\Delta^2} \right)$ cascades to obtain enough valid samples for all the edges in the graph.
	\proof We notice that if we want to learn all edges in the graph, it implies that we have to learn all the edges from the $N-4$ nodes to node 1. However, if $i$ is not part of the clique, any valid sample for such an edge (i, 1) has to have $i$ as its source. Having enough valid samples for each of these edges is therefore equivalent to collecting $\Omega(\frac{1}{\Delta^2})$ copies of $N-4$ distinct coupons in the standard Coupon collector problem. Using results from \cite{newman1960double, erdHos1961classical}, we need $\Omega(\left(K \log(K) + (d-1) \cdot K \cdot \log\log(K) \right)$ samples to obtain $d$ copies of each coupon when there are $K$ distinct coupons in total, which is here $\Omega\left((N-4) \log(N-4) + (\frac{1}{\Delta^2} - 1) \cdot (N-4) \cdot \log\log(N-4) \right)$ cascades. Using standard approximation, we get the desired result.
\end{claim}

Combining the results:

\begin{theorem}
	We need at least $\Omega\left(N \log(N) + \frac{N \log\log(N)}{\Delta^2} \right)$ cascades to learn any mixture of directed graphs of minimum out-degree 3.
\end{theorem}

\subsection{Undirected lower bound}
We reuse a lot of the techniques in the previous subsection. This time, we consider a simple line graph on $N$ nodes, where for all $1 \leq i \leq N-1$, node $i$ is connected to node $i+1$. Like in the previous example, the weights are all $p$ in $G_1$, and all $p + \Delta$ in $G_2$.

Reusing Claim \ref{cl:validSamples}, we now prove:

\begin{claim}
	We need at least $\Omega\left( \frac{N}{\Delta^2} \right)$ cascades to obtain enough valid samples for edge (1,2).
	\proof To provide a valid sample, either:
	\begin{itemize}
		\item Node $1$ is the source, which happens with probability $\mathcal{P}_1 = \frac{1}{N}$.
		\item Node 2 was infected, which happens with probability $\mathcal{P}_2 \leq \displaystyle\sum_{i=2}^{N} \frac{1}{N} p_{max}^{i-2} \leq \frac{1}{N} \frac{1}{1-p_{max}}$.
	\end{itemize}
Therefore, the probability of getting a valid sample is smaller than $\mathcal{P}_1 + \mathcal{P}_2  \leq \frac{1}{N} \cdot \frac{2}{1- p_{max}}$. Hence, we need at least $\Omega (\frac{1-p_{max}}{2} \cdot N \cdot \frac{1}{\Delta^2}) = \Omega\left( \frac{N}{\Delta^2} \right)$ cascades to obtain enough valid samples.

\end{claim}

Since we need to learn at least edge $(1,2)$ to learn all the edges of this graph:

\begin{theorem}
	We need at least $\Omega\left(\frac{N}{\Delta^2} \right)$ cascades to learn any mixture of undirected graphs.
\end{theorem}

\section{Directed graphs}
\subsection{Structures}
\begin{figure}
	\centering
	\begin{subfigure}[t]{0.28\textwidth}
		\centering
		\begin{tikzpicture}
		\node[main node, fill=white] (0) {$a$};
		\node[main node, fill=pink] (1) [right = 0.5cm of 0] {$u$};
		\node[main node, fill=white] (2) [above right = 0.5cm and 0.2cm of 1]  {$b$};
		\node[main node, fill=white] (3) [below right = 0.5cm and 0.2cm of 1] {$c$};
		
		\path[draw,thick,->]
		(1) edge node {} (0)
		(1) edge node {} (2)
		(1) edge node {} (3);
		\end{tikzpicture}
		\caption{A star vertex $u$ for a directed graph.}\label{fig:stardirected}
	\end{subfigure}\quad
	\begin{subfigure}[t]{0.28\textwidth}
		\centering
		\begin{tikzpicture}
		\node[main node, fill=pink] (0) {$a$};
		\node[main node, fill=white] (1) [right = 0.5cm of 0] {$u$};
		\node[main node, fill=white] (2) [above right = 0.5cm and 0.2cm of 1]  {$b$};
		\node[empty node] (3) [below right = 0.5cm and 0.2cm of 1] {};
		
		\path[draw,thick, ->]
		(0) edge node {} (1)
		(1) edge node {} (2);
		\end{tikzpicture}
		\caption{First structure to ensure sign consistency.}\label{fig:2nodes}
	\end{subfigure}\quad
	\begin{subfigure}[t]{0.28\textwidth}
		\centering
		\begin{tikzpicture}
		\node[main node, fill=pink] (0) {$a$};
		\node[main node, fill=white] (1) [right = 0.5cm of 0] {$u$};
		\node[main node, fill=white] (2) [above right = 0.5cm and 0.2cm of 1]  {$b$};
		\node[empty node] (3) [below right = 0.5cm and 0.2cm of 1] {};
		
		\path[draw,thick, ->]
		(0) edge node {} (1)
		(1) edge node {} (2)
		(0) edge node {} (2);
		\end{tikzpicture}
		\caption{Second structure to ensure sign consistency.}\label{fig:triangle}
	\end{subfigure}
	\caption{Structures for directed graphs of minimum out-degree three.}\label{fig:directed}
\end{figure}
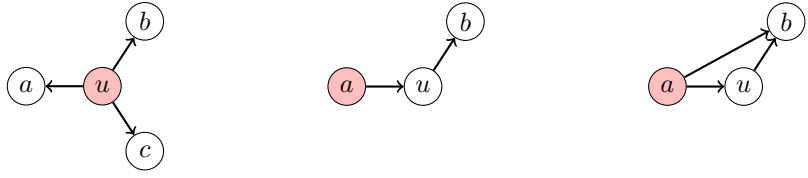

\paragraph{Star vertex} For directed graph of out-degree at least 3, every vertex is a star vertex. This implies we can reuse the star vertex equations to learn the weights of the whole neighborhood of each node. However, if we learn the neighborhoods of node $u$ in both graphs, which we call $\mathcal{N}_1^u$ and $\mathcal{N}_2^u$, as well as the neighbordhoods of node $a$, which we call $\mathcal{N}_1^a$ and $\mathcal{N}_2^a$, it is impossible to recover from the star structure alone if $\mathcal{N}_1^u$ and $\mathcal{N}_1^a$ are in the same mixture, or if it is $\mathcal{N}_1^u$ and $\mathcal{N}_2^a$ instead. We therefore use the two other structures in Figure \ref{fig:directed} to ensure mixture consistency.

\paragraph{Mixture consistency} Suppose we have learned the weights of all the edges stemming from $a$, as well as all the weighted edges stemming from $u$, and suppose there is no edge between $a$ and $b$. The probability that $a$ infected $u$, which in turn infected $b$ is:
$$ \mathbb{P}(a \rightarrow u \rightarrow b |  a \in I_0) = \frac{p_{au} p_{ub} + q_{au}q_{ub}}{2}. $$
This gives us a way to decide whether $\mathcal{N}_1^u$ and $\mathcal{N}_1^a$ are in the same mixture, or if it is $\mathcal{N}_1^u$ and $\mathcal{N}_2^a$ instead. Indeed, if we know $p_{au} \in \mathcal{N}_1^a, q_{au} \in \mathcal{N}_2^a$, and we also know $w_{ub} \in \mathcal{N}_1^u, w'_{ub} \in \mathcal{N}_2^u$, and we have an estimator $\hat{Y}_{au, ub}$ for $\mathbb{P}(a \rightarrow u \rightarrow b |  a \in I_0)$, then we can check whether $\hat{Y}_{au, ub} \approx \frac{p_{au} w_{ub} + q_{au}w'_{ub}}{2}$, in which case $\mathcal{N}_1^u$ belongs with $\mathcal{N}_1^a$, or whether  $\hat{Y}_{au, ub} \approx \frac{p_{au} w'_{ub} + q_{au}w_{ub}}{2}$, in which case $\mathcal{N}_2^u$ belongs in the with $\mathcal{N}_1^a$. We call this procedure \textsc{CheckPath}.

Similarly, if there is an edge between $a$ and $b$, then:
$$ \mathbb{P}(a \rightarrow u \rightarrow b |  a \in I_0) = \frac{p_{au} (1 - p_{ab}) p_{ub} + q_{au}(1-q_{ab})q_{ub}}{2}. $$
This also allows us to ensure mixture consistency. We call this procedure \textsc{CheckTriangle}. 
%\begin{algorithm}
%	\caption{\textsc{Check2Edges}}\label{app:algodirected}
%	\hspace*{\algorithmicindent} \textbf{Input} Vertices $a,u$, set of already learned weights $W$, neighborhoods of $u$  $\mathcal{N}_1^u$ and $\mathcal{N}_2^u$ \\
%	\hspace*{\algorithmicindent} \textbf{Output} Whether or not $\mathcal{N}_1^u$ belongs in the first component of $W$.
%	\begin{algorithmic}[1]
%		\State $W = \{W_1, W_2 \}$
%		\State $p_{au}$  (resp. $q_{au}$) $\leftarrow$ weight of edge $(a, u)$ in $W_1$ (resp. $W_2$)
%		\State Select $b \neq a$ neighbor of $u$ \Comment{\textbf{$b$ exists as $u$ is of out-degree at least 3}}
%		\State  weight of edge $(a, u)$ in $W_2$
%		\State Select any first node $u$
%		\State $W \leftarrow \textsc{LearnStar}(u, E, W)$
%		\State $S = \{u\}$
%		\While {$S \neq V$}
%		\State Select $u \in S, v \in V \backslash S$ such that $(u,v) \in E$ 
%		\Comment{\textbf{$v$ has out-degree at least 3}}
%		\State $W_1, W_2 \leftarrow  \textsc{LearnStar}(v, E, W)$
%		\If{$(u,v) \notin E$} \Comment{\textbf{Use first structure.}}
%		\Else \Comment{\textbf{Use second structure.}}
%		\EndIf
%		\State $S \leftarrow S \cup \{v\}$
%		\EndWhile
%		% \EndIf
%		\State {\bf Return} $W$
%	\end{algorithmic}
%\end{algorithm}

Here is the final algorithm:

\begin{algorithm}[H]
	\caption{Learn the weights of directed edges}\label{app:algodirected}
	\hspace*{\algorithmicindent} \textbf{Input} Vertex set $V$\\
	\hspace*{\algorithmicindent} \textbf{Output} Edge weights for the two epidemics graphs
	\begin{algorithmic}[1]
		\State $E \leftarrow \textsc{LearnEdges}(V)$ 
		\State Select any first node $v$
		\State $W \leftarrow \textsc{LearnStar}(v, E, W)$
		\State $S = \{v\}$
		\While {$S \neq V$}
		\State Select $a \in S, v \in V \backslash S$ such that $(a,u) \in E$ 
		\Comment{\textbf{$v$ has out-degree at least 3}}
		\State $\mathcal{N}_1, \mathcal{N}_2 \leftarrow  \textsc{LearnStar}(u, E, W)$
		\State Select $b \neq a$ neighbor of $u$ \Comment{\textbf{$b$ exists because u os of degree at least 3.}}
		\If{$(a,b) \notin E$} \Comment{\textbf{Use first structure.}}
		\If{\textsc{CheckPath}$(v,u,b, W, \mathcal{N}_1, \mathcal{N}_2)$}
		\State $W = \{W_1 \cup \mathcal{N}_1, W_2 \cup \mathcal{N}_2\}$
		\Else
		\State $W = \{W_1 \cup \mathcal{N}_2, W_2 \cup \mathcal{N}_1\}$
		\EndIf
		\Else \Comment{\textbf{Use second structure.}}
		\If{\textsc{CheckTriangle}$(v,u,b, W, \mathcal{N}_1, \mathcal{N}_2)$}
		\State $W = \{W_1 \cup \mathcal{N}_1, W_2 \cup \mathcal{N}_2\}$
		\Else
		\State $W = \{W_1 \cup \mathcal{N}_2, W_2 \cup \mathcal{N}_1\}$
		\EndIf
		\EndIf
		\State $S \leftarrow S \cup \{u\}$
		\EndWhile
		% \EndIf
		\Return $W$
	\end{algorithmic}
\end{algorithm}

\section{Unbalanced/Unknown Mixtures} \label{app:unknownAlpha}

In this section we provide the primitives required for \textsc{LearnStar} and \textsc{LearnLine}, when the first mixture occurs with probability  $\alpha$ and the second mixture with probability $(1-\alpha)$. 

\textbf{Notations:} In this section, to avoid clutter in notation we use $i$, $j$ and $k$ to be all \emph{distinct} unless mentioned otherwise. Also, let $\sigma(\{a,b,c\}) =\{ (a,b,c), (b,c,a), (c, a, b)\}$ denote all the permutations of $a$, $b$, and $c$.

\begin{claim} \label{cl:genEstimators}
	If $a$ and $b$ are two distinct nodes of $V_1 \cap V_2$ such that $(a,b) \in E_1 \cap E_2$ then  under general mixture model $X_{ab} = \alpha p_{ab} + (1- \alpha) q_{ab}$. 
	
	Fuither, when the four nodes $u$, $a$, $b$ and $c$ forms a \underline{star graph} (Fig.~\ref{fig:star}) with $u$ in the center under general mixture model 
	\begin{align*}
	&1)~\forall i, j \in\{a,b,c\}, i, j\neq u, Y_{ui, uj} = \alpha  p_{ui}p_{uj} + (1- \alpha)  q_{ui}q_{uj},\\
	& 2)~ Z_{ua, ub, uc} = \alpha p_{ua}p_{ub}p_{uc} + (1-\alpha )q_{ua}q_{ub}q_{uc}.
	\end{align*}
	
	Finally,  when the four nodes $u$, $a$, $b$ and $c$ forms a \underline{line graph} (Fig.~\ref{fig:line}) under general mixture model 
	\begin{gather*}
	1)~Y^|_{ua, ub} = \alpha  p_{ua}p_{ub} + (1- \alpha)  q_{ua}q_{ub},
	2)~Y^|_{ub, bc} = \alpha  p_{ub}p_{bc} + (1- \alpha)  q_{ub}q_{bc},\\
	3)~Z^|_{ua, ub, bc} = \alpha p_{ua}p_{ub}p_{bc} + (1-\alpha )q_{ua}q_{ub}q_{bc}.
	\end{gather*}
\end{claim}
The proof of the above claim is omitted as it follows closely the proofs of Claim~\ref{cl:trivialXab}, \ref{claim:eststar}, and \ref{claim:estline}.

\subsection{Star Graph}

We now present the following two lemmas which recover  the weights $p_{ui}$,and  $q_{ui}$ for all $ i\in \{a,b,c\}$ in the star graph (Fig.~\ref{fig:star}), and the general mixture parameter $\alpha$, respectively.

\begin{lemma}[Weights of General Star Graph]\label{lem:stargen}
	Under Conditions \ref{ass:necessary} and  \ref{ass:main}, in the setting of infinite samples, for the $star$ structure $(u,a,b,c)$ with $u$ as the central vertex the weight of any edge $(u,a)$  is given by:
	\begin{align*}
	p_{ua} &= X_{ua} + s_{ua}\sqrt{\tfrac{1-\alpha}{\alpha}} \sqrt{\frac{(Y_{ua,ub} - X_{ua}X_{ub})(Y_{ua,uc} - X_{ua}X_{uc})}{Y_{ub,uc} - X_{ub}X_{uc}}}\\ 
	q_{ua} &= X_{ua} - s_{ua} \sqrt{\tfrac{\alpha}{1-\alpha}} \sqrt{\frac{(Y_{ua,ub} - X_{ua}X_{ub})(Y_{ua,uc} - X_{ua}X_{uc})}{Y_{ub,uc} - X_{ub}X_{uc}}}
	\end{align*} 
	where $s_{ua} \in \{-1,1\}$ and $b,c \in N_1(u) \cap N_2(u)$ such that $b,c \ne a$, $b\neq c$.
\end{lemma}
\begin{proof}
	We notice that for $r \neq j \in \{a,b,c\}$
	\begin{align*}
	\left(Y_{ui, uj} - X_{ui}X_{uj}\right) &= \left(\alpha p_{ui}p_{uj} + (1-\alpha) q_{ui}q_{uj}\right) -\left(  \alpha p_{ui} + (1-\alpha) q_{ui}\right)\left( \alpha p_{uj} + (1-\alpha) q_{uj}\right) \\
	&=  \alpha(1-\alpha)(p_{ui} - q_{ui})(p_{uj} - q_{uj}).
	\end{align*}
	The rest of the proof follows the same steps as given in the proof of Lemma~\ref{lem:star} with the above modification.
\end{proof}

\begin{lemma}[Sign Ambiguity Star Graph]\label{lem:starSign}
	Under Conditions \ref{ass:necessary} and  \ref{ass:main}, in the setting of infinite samples, for edges $(u,a), (u,b)$ for the $star$ structure $(u,a,b,c)$ with $u$ as the central vertex, the sign pattern $s_{ua}, s_{ub}$ satisfy the following relation.
	\[
	s_{ub}s_{ua} = \mathsf{sgn}(Y_{ua,ub} - X_{ua}X_{ub}).
	\]
\end{lemma}
\begin{proof}
	The proof of the first statement follows the same logic as the proof of Lemma~\ref{lem:sign}, after noting that $\mathsf{sgn}(\alpha(1-\alpha)) = 1$ for $\alpha \in (0,1)$. 
\end{proof}

\subsection{Line Graph}
We now present the recovery of parameters in the case of a line graph with knowledge of $\alpha$

\begin{lemma}[Weights of General Line Graph]\label{lem:linegen}
		Under Conditions \ref{ass:necessary} and  \ref{ass:main}, in the setting of infinite samples, the weights of the edges $(u,a)$, and $(u,b)$ for a $line$ graph $a-u-b-c$ can be learned in closed form (as given in the proof), as a function of\\
	(1) the mixture parameter $\alpha$,  \\
	(2) estimators $X_{ua}$, $X_{ub}$, $X_{bc}$, $Y^|_{ua,ub}$, $Y^|_{ub,bc}$, and $Z^|_{ua,ub,bc}$,\\ 
	(3) one variable $s_{ub} \in \{-1,+1\}$.
\end{lemma}
\begin{proof}
	We first note that we have access to the following three relations 
	\begin{align*}
	1)\quad&(Y^|_{ua,ub} - X_{ua}X_{ub}) = \alpha(1-\alpha) (p_{ua} - q_{ua})(p_{ub} - q_{ub})\\
	2)\quad&(Y^|_{ub,bc} - X_{ub}X_{bc}) = \alpha(1-\alpha) (p_{ub} - q_{ub})(p_{bc} - q_{bc})\\
	3)\quad&(Z^|_{ua, ub, bc} + X_{ua}X_{ub}X_{bc} - X_{ua}Y^|_{ub, bc} - X_{bc}Y^|_{ua,ub})\\
	&= \alpha(1-\alpha)((1 - \alpha)p_{ub} + \alpha q_{ub})(p_{ua} - q_{ua})(p_{bc} - q_{bc}).
	\end{align*} 
	The first two inequalities follow similar to Lemma~\ref{lem:line}. We derive the final equality below.
	\begin{align*}
	&Z^|_{ua, ub, bc} + X_{ua}X_{ub}X_{bc} - X_{ua}Y^|_{ub, bc} - X_{bc}Y^|_{ua,ub}\\
	&= \alpha p_{ua}p_{ub}p_{bc} + (1- \alpha)q_{ua}q_{ub}q_{bc}\\
	&- (\alpha p_{ua} + (1 - \alpha)q_{ua})((\alpha p_{ub}p_{bc} + (1 - \alpha)q_{ub}q_{bc}) - (\alpha p_{bc} + (1 - \alpha)q_{bc})((\alpha p_{ua}p_{ub} + (1 - \alpha)q_{ua}q_{ub})\\
	& + (\alpha p_{ua} + (1 - \alpha)q_{ua})(\alpha p_{ub} + (1 - \alpha)q_{ub})(\alpha p_{bc} + (1 - \alpha)q_{bc})\\
	&=  \alpha (1 - \alpha)^2 p_{ua}p_{ub}p_{bc} + \alpha^2(1- \alpha)q_{ua}q_{ub}q_{bc}\\
	&- \alpha(1-\alpha)^2p_{ua}p_{ub}q_{bc}  + \alpha^2(1- \alpha)p_{ua}q_{ub}p_{bc} 
	- \alpha(1-\alpha)^2q_{ua}p_{ub}p_{bc}\\
	&- \alpha^2(1-\alpha)q_{ua}q_{ub}p_{bc} + \alpha(1-\alpha)^2q_{ua}p_{ub}q_{bc}  -\alpha^2(1-\alpha)p_{ua}q_{ub}q_{bc}\\
	&= \alpha(1-\alpha)((1 - \alpha)p_{ub} + \alpha q_{ub})(p_{ua} - q_{ua})(p_{bc} - q_{bc})
	\end{align*} 
	
	Therefore, we obtain the following quadratic equation in $p_{ub}$ and $q_{ub}$ (unlike the $\alpha = 1/2$ case it cannot be easily reduced to a linear equation),
	\[
	\frac{\alpha(1-\alpha)(p_{ub} - q_{ub})^2}{((1 - \alpha)p_{ub} + \alpha q_{ub})}= \frac{(Y^|_{ua,ub} - X_{ua}X_{ub})(Y^|_{ub,bc} - X_{ub}X_{bc})}{(Z^|_{ua, ub, bc} + X_{ua}X_{ub}X_{bc} - X_{ua}Y^|_{ub, bc} - X_{bc}Y^|_{ua,ub})} := C_{ub}^|
	\]
	Note that $X_{ub} = \alpha p_{ub} + (1- \alpha)q_{ub}$, thus the above can be reduced to 
	\begin{align*}
	& \frac{\alpha(1-\alpha)(p_{ub}- X_{ub})^2/(1-\alpha)^2}{(p_{ub}(1- 2\alpha)+ \alpha X_{ub})/(1-\alpha)} = C_{ub}^|\\
	& p_{ub}^2 - 2\left(X_{ub}+ \tfrac{(1-2\alpha)}{2\alpha}C_{ub}^|\right)p_{ub} = C_{ub}^| X_{ub} - X^2_{ub}\\
	&p_{ub} = X_{ub} + \tfrac{(1-2\alpha)}{2\alpha} C_{ub}^| 
	+s_{ub}\sqrt{\left(\tfrac{(1-2\alpha)}{2\alpha} C_{ub}^| \right)^2 + \tfrac{1-\alpha}{\alpha}C_{ub}^| X_{ub} }\\
	& q_{ub} = X_{ub} - \tfrac{(1-2\alpha)}{2(1-\alpha)} C_{ub}^| 
	- s_{ub}\sqrt{\left(\tfrac{(1-2\alpha)}{2(1-\alpha)} C_{ub}^| \right)^2 + \tfrac{\alpha}{1-\alpha}C_{ub}^| X_{ub} }
	\end{align*}
	We substitute in the above two equations $\theta$ and $s_{\alpha}$ as defined below
	$$\alpha = \tfrac{1}{2}(1 - s_{\alpha}\sqrt{\theta}), 
	\qquad (1-\alpha) = \tfrac{1}{2}(1 + s_{\alpha}\sqrt{\theta}),
	\qquad (1-2\alpha) = s_{\alpha}\sqrt{\theta}.$$ 
	From the substitution we obtain,
	\begin{align*}
	&p_{ub} = X_{ub} +
	\tfrac{s_{\alpha}\sqrt{\theta}(1 + s_{\alpha}\sqrt{\theta})C^|_{ub}}{(1-\theta)}\left( 1 + s_{\alpha} s_{ub} \sqrt{1+ \tfrac{(1-\theta)X_{ub}}{\theta C^|_{ub}}} \right)\\
	&q_{ub} = X_{ub} -
	\tfrac{s_{\alpha}\sqrt{\theta}(1 - s_{\alpha}\sqrt{\theta})C^|_{ub}}{(1-\theta)}\left( 1 + s_{\alpha} s_{ub} \sqrt{1+ \tfrac{(1-\theta)X_{ub}}{ \theta C^|_{ub}}} \right)
	\end{align*}
	
	Next we use $p_{ub}$, and $q_{ub}$ to obtain $p_{ua}$, and $q_{ua}$. Specifically, we have
	\begin{align*}
	&\alpha(1-\alpha)(p_{ub} - q_{ub})(p_{ua} - q_{ua}) = (Y_{ua, ub}^| - X_{ua} X_{ub})\\
	&(p_{ua} - q_{ua})  = \frac{4(Y_{ua, ub}^| - X_{ua} X_{ub})}{ s_{\alpha} \sqrt{\theta}\left( 1 + s_{\alpha} s_{ub} \sqrt{1+ \tfrac{(1-\theta)X_{ub}}{ \theta C^|_{ub}}} \right)}.
	\end{align*}

	Finally, we use the above relation to arrive at the required result.
	\begin{align*}
	&p_{ua} = X_{ua} +\frac{2(1+s_{\alpha}\sqrt{\theta}) (Y_{ua, ub}^| - X_{ua} X_{ub})}{ s_{\alpha} \sqrt{\theta}\left( 1 + s_{\alpha} s_{ub} \sqrt{1+ \tfrac{(1-\theta)X_{ub}}{ \theta C^|_{ub}}} \right)}\\
	&q_{ua} = X_{ua} - \frac{2(1-s_{\alpha}\sqrt{\theta}) (Y_{ua, ub}^| - X_{ua} X_{ub})}{ s_{\alpha} \sqrt{\theta}\left( 1 + s_{\alpha} s_{ub} \sqrt{1+ \tfrac{(1-\theta)X_{ub}}{ \theta C^|_{ub}}} \right)}
	\end{align*}
\end{proof}

\begin{lemma}[Sign Ambiguity Line graph on $5$ nodes]\label{lem:lineSign}
	Under Conditions \ref{ass:necessary} and  \ref{ass:main}, in the setting of infinite samples, for a line structure $a-u-b-c-d$ the sign patterns  $s_{ub}$ and $s_{bc}$ satisfy the relation, $s_{ub} s_{bc} = \mathsf{sgn}(Y^|_{ub,bc} - X_{ub}X_{bc})$. 
\end{lemma}
\begin{proof}
	The proof is almost identical to the other sign ambiguity proofs.
\end{proof}

\subsection{Finite Sample Complexity}

We start by observing that the Claim~\ref{cl:Mu} still holds in the general case.
\begin{claim}\label{cl:epsilongeneral}
	If we can estimate $X_{ua}, Y^*_{ua,ub}, Y^|_{ua,ab}$ and $Z^|_{ua,ub,bc}$ within $\epsilon_1$, we can estimate $p_{ua}$  and $q_{ua}$ within precision 
	$\epsilon = \mathcal{O}\left( \epsilon_1 / \min(p_{min}, \Delta)^5 \min(\alpha, 1- \alpha)^4\right) $.
\end{claim}
\begin{proof}
	The proof proceeds in a very similar manner as Claim\ref{cl:epsilon2}. Following the derivations for $\hat{p}_{ua}$ and $\hat{q}_{ua}$ in the proof of Claim\ref{cl:epsilon2},  we can see that for the star primitive all the computation carry over with a scaling of $\tfrac{4}{\alpha(1-\alpha)}$ as we have 
	$Y^*_{ur,us} - X_{ur}X_{us} \geq \Delta^2 \alpha(1-\alpha)$ instead of $\Delta^2 /4$. 
	
	The line primitive presents with increased difficulty as the estimator is more complex. We first observe that $ \alpha(1-\alpha) \Delta^2 \leq C_{ub}^| \leq \max(\alpha, (1-\alpha))$. 
	We recall that 
	\begin{align*}
	&(Z^|_{ua, ub, bc} + X_{ua}X_{ub}X_{bc} - X_{ua}Y^|_{ub, bc} - X_{bc}Y^|_{ua,ub}) \\
	&= \alpha(1-\alpha) ((1 - \alpha)p_{ub} + \alpha q_{ub})(p_{ua} - q_{ua})(p_{bc} - q_{bc})\\
	&\geq \min(\alpha, 1-\alpha)^2p_{min} \min(p_{min}, \Delta)^2/2,\\
	& (Y^|_{ua,ub} - X_{ua}X_{ub}) = \alpha(1-\alpha) (p_{ua} - q_{ua})(p_{ub} - q_{ub})
	\geq  \min(\alpha, 1-\alpha) \min(p_{min}, \Delta)^2/2.
	\end{align*}
	Let us assume the error in $(Z^|_{ua, ub, bc} + X_{ua}X_{ub}X_{bc} - X_{ua}Y^|_{ub, bc} - X_{bc}Y^|_{ua,ub})$ is bounded as $\epsilon_d$ and the error in $(Y^|_{ua,ub} - X_{ua}X_{ub}) (Y^|_{ub,bc} - X_{ub}X_{bc})$ is bounded as $\epsilon_n$. We have $\epsilon_n \leq 4\epsilon_1$ and $\epsilon_d \leq 3\epsilon_1$ as all the estimators are assumed to have error bounded by $\epsilon_1$.

	Therefore, using  $|x/y - \hat{x}/\hat{y}| \leq x/y(\delta_x/x  +\delta_y/y)+ \mathcal{O}(\delta_x\delta_y)$,
	\begin{align*}
	| \hat{C}_{ub}^|  - C_{ub}^| | \leq \epsilon_c &:=  \mathcal{O}\left(\tfrac{\epsilon_n}{  \min(\alpha, 1-\alpha)^2  \min(p_{min}, \Delta)^4}
	+  \tfrac{\epsilon_d} { \min(\alpha, 1-\alpha)^2 p_{min} \min(p_{min}, \Delta)^2}  \right) \\
	&= \mathcal{O}(\epsilon_1/ \min(\alpha, 1-\alpha)^2  \min(p_{min}, \Delta)^4).
	\end{align*}
	
	Using the above bound in the expression of $p_{ua}$ we can obtain, 
	\begin{align*}
	&| \hat{p}_{ua} - p_{ua}| 
	\leq | \hat{X}_{ua} - X_{ua}| + \tfrac{(1-2\alpha)}{2\alpha} |\hat{C}^|_{ua} - C^|_{ua}| + \dots\\ 
	&+ \big\lvert \sqrt{\left(\tfrac{(1-2\alpha)}{2\alpha} \hat{C}^|_{ua}\right)^2+\tfrac{(1-\alpha)}{\alpha} \hat{C}^|_{ua} \hat{X}_{ua}}
	-  \sqrt{\left(\tfrac{(1-2\alpha)}{2\alpha} C^|_{ua}\right)^2+\tfrac{(1-\alpha)}{\alpha} C^|_{ua} X_{ua}} \big\rvert\\
	&\leq | \hat{X}_{ua} - X_{ua}| + \tfrac{(1-2\alpha)}{2\alpha} |\hat{C}^|_{ua} - C^|_{ua}| + \dots \\
	&+\frac{\left(\tfrac{(1-2\alpha)}{2\alpha}\right)^2 |\hat{C}^|_{ua} - C^|_{ua}| (\hat{C}^|_{ua}  + C^|_{ua}) + \tfrac{(1-\alpha)}{\alpha} | \hat{C}^|_{ua} \hat{X}_{ua} - C^|_{ua} X_{ua}|}
	{\sqrt{\left(\tfrac{(1-2\alpha)}{2\alpha} C^|_{ua}\right)^2+\tfrac{(1-\alpha)}{\alpha} C^|_{ua} X_{ua}} }\\
	& \leq \epsilon_1 + \tfrac{(1-2\alpha)}{2\alpha} \epsilon_c + \tfrac{2}{(1-\alpha)\min(p_{min}, \Delta)} \left(2 \left(\tfrac{(1-2\alpha)}{2\alpha}\right)^2 \epsilon_c  + \tfrac{(1-\alpha)}{\alpha} (\epsilon_1 + \epsilon_c)\right) + o(\epsilon_1) + o(\epsilon_c)\\
	&\leq \mathcal{O}(\epsilon_1/\min(p_{min}, \Delta) \alpha (1-\alpha)) + \mathcal{O}( \epsilon_c/\min(p_{min}, \Delta) \alpha^2 (1-\alpha)) + o(\epsilon_1) + o(\epsilon_c)
	\end{align*}
	Therefore, using the estimate of $\epsilon_c$ we obtain, 
	$$| \hat{p}_{ua} - p_{ua}| \leq \mathcal{O}\left( \epsilon_1 / \min(p_{min}, \Delta)^5 \alpha \min(\alpha, 1- \alpha)^3\right).$$
	Switching $\alpha$ and $(1-\alpha)$ gives us the same bounds for $|\hat{q}_{ua} - q_{ua}|$. 
\end{proof}
In the above derivation we have used $\sqrt{\left(\tfrac{(1-2\alpha)}{2\alpha} C^|_{ua}\right)^2+\tfrac{(1-\alpha)}{\alpha} C^|_{ua} X_{ua}} \geq (1-\alpha)\min(p_{min}, \Delta) / 2$. We now derive the above inequality. 
\begin{align*}
&\lvert \sqrt{\left(\tfrac{(1-2\alpha)}{2\alpha} C^|_{ua}\right)^2+\tfrac{(1-\alpha)}{\alpha} C^|_{ua} X_{ua}} \rvert
= | p_{ua} - X_{ua} - \tfrac{(1-2\alpha)}{2\alpha} C^|_{ua}|\\
&= \lvert (1-\alpha) (p_{ua} - q_{ua}) - \tfrac{(1-2\alpha)(1-\alpha) (p_{ua} - q_{ua})^2 }{2 ((1-\alpha)p_{ua}+\alpha q_{ua})  }\rvert \\
&\geq \begin{cases*}
(1-\alpha) \min(p_{min}, \Delta), (\alpha \geq 1/2 \wedge p_{ua} \geq q_{ua}) \vee (\alpha < 1/2 \wedge p_{ua} < q_{ua})\\
(1-\alpha) \min(p_{min}, \Delta)|1 - \tfrac{(1-2\alpha)}{2(1-\alpha)}|, (\alpha < 1/2 \wedge p_{ua} \geq q_{ua}) \\
(1-\alpha) \min(p_{min}, \Delta) |1 - \tfrac{(2\alpha-1)}{2\alpha}|, (\alpha \geq 1/2 \wedge p_{ua} < q_{ua}), 
\end{cases*}
\end{align*}

Finally, using union bound on all the estimators involved accross all possible edges, we can obtain the error bound in the following Theorem~\ref{thm:mainexistgen}. 
\begin{theorem}\label{thm:mainexistgen}
	Suppose Condition \ref{ass:necessary} and \ref{ass:main} are true, there exists an algorithm that runs on epidemic cascades over a mixture of two undirected, weighted graphs $G_1 = (V, E_1)$ and $G_2 = (V, E_2)$, and recovers the edge weights corresponding to each graph up to precision $\epsilon$ in time $O(N^2)$ and sample complexity $O\left(\frac{N\log N}{\epsilon^2\Delta^4}\right)$ for $\alpha = 1/2$ and 
	$O\left(\frac{N\log N}{\epsilon^2\Delta^{10} \min(\alpha, 1-\alpha)^8 }\right)$ 
	for general $\alpha\in (0,1), \alpha \neq 1/2$, 
	where $N = |V|$.
\end{theorem}

\end{document}